\documentclass[conference]{IEEEtran}
\input{datalabmacros}

\usepackage{cite}
\usepackage{amsmath,amssymb,amsfonts}
\usepackage{algorithmic}
\usepackage{graphicx}
\usepackage{textcomp}
\usepackage{xcolor}
\def\BibTeX{{\rm B\kern-.05em{\sc i\kern-.025em b}\kern-.08em
    T\kern-.1667em\lower.7ex\hbox{E}\kern-.125emX}}
\begin{document}

\title{\HND: From Truth Discovery to Ability Discovery by Recovering Matrices with the Consecutive Ones Property
}

\author{\IEEEauthorblockN{Zixuan Chen}
\IEEEauthorblockA{\textit{Northeastern University}\\
chen.zixu@northeastern.edu}
\and
\IEEEauthorblockN{Subhodeep Mitra}
\IEEEauthorblockA{\textit{Google}\\
mitradeep1@gmail.com}
\and
\IEEEauthorblockN{R. Ravi}
\IEEEauthorblockA{\textit{Carnegie Mellon University}\\
ravi@andrew.cmu.edu}
\and
\IEEEauthorblockN{Wolfgang Gatterbauer}
\IEEEauthorblockA{\textit{Northeastern University}\\
w.gatterbauer@northeastern.edu}
}

\maketitle

\thispagestyle{plain}
\pagestyle{plain}

\begin{abstract}
We analyze a general problem in a crowd-sourced setting where 
one user asks a question (also called item) and other users return answers (also called labels) for this question.
Different from existing crowd sourcing work which focuses on \emph{finding the most appropriate label for the question (the ``truth'')}, our problem is to \emph{determine a ranking of the users based on their ability to answer questions}.
We call this problem ``\emph{ability discovery}'' to emphasize the connection to and duality with the more well-studied 
problem of ``truth discovery''.

To model items
and their labels in a principled way, we draw upon Item Response Theory (IRT) which is the widely accepted theory behind standardized tests such as SAT and GRE.
We start from an idealized setting where the relative performance of users is \emph{consistent} across items and better users choose better fitting labels for \emph{each} item.
We posit that a principled algorithmic solution to our more general problem should solve this ideal setting correctly
and observe that the response matrices 
in this setting
obey the \emph{Consecutive Ones Property} (C1P).
While C1P is well understood algorithmically with various discrete algorithms, 
we devise a novel 
variant of 
the HITS algorithm which we call ``{\HND}'' (or \HnD), and prove that it can recover the ideal C1P-permutation in case it exists.
Unlike fast combinatorial algorithms for finding the consecutive ones permutation (if it exists), \HnD also returns an ordering when such a permutation does not exist.
Thus it provides a \emph{principled heuristic for our problem that is guaranteed
to return the correct answer in the ideal setting}.
Our experiments show that \HnD produces user rankings with robustly high accuracy compared to state-of-the-art truth discovery methods. 
We also show that our novel variant of \HITS scales better in the number of users than \ABH, the only prior spectral C1P reconstruction algorithm.

\end{abstract}


\section{Introduction}
\label{sec:intro}

\introparagraph{Motivation}
We first present a couple of examples from a class to a more general crowdsourcing context.

\begin{example}[Student ranking]
Kiyana, an innovative instructor for an online class who suffered from the leakage of previous exam questions and difficulty of creating new ones, notices a lot of interactions in student forums like Piazza~\cite{piazza} and pilots a new learning approach.
Since students are willing to ask and answer questions, Kiyana aims to utilize such communication for both practice and assessment of the class by requiring students to suggest and answer multiple-choice questions (MCQs) themselves.
The task of students is to come up with meaningful MCQs (including question stems and choices) 
related to the topics of the class and answer the questions from others.
In this way, the initiative of and interactions between students are encouraged, and their performances can be used as another important measure (e.g., a ``participation'' score) towards the grade assessment, in addition to traditional exam scores.
To assess students, Kiyana first simply counts how many times a student answers questions, which is the traditional way for an instructor to give a participation score for students in a forum.
However, in this way, the grades are biased towards students who answer a lot of questions randomly.
The second attempt is to require all students to answer the same number of questions and rank them by how many questions they answer correctly, which still requires a lot of efforts for her to figure out all the correct answers and suffers from the problem that each question has quite different difficulties.
Kiyana wonders whether there is a more principled way for ranking students by their abilities to answer questions correctly.
\end{example}

\begin{example}[Crowd workers ranking]
Daiyu wants to release a human intelligence task which consists of a set of questions at a crowdsourcing platform like Amazon Mechanical Turk~\cite{amazonturk}.
Suffering from low-quality answers from the crowd workers, she wonders whether there is a better way to select top crowd workers instead of simply setting thresholds for the number of tasks they have finished or finished successfully.
\end{example} 

The above examples motivate our problem of ``ability discovery'' 
which ranks the users (students/workers) based on their abilities to answer questions correctly.

\introparagraph{The ability discovery problem}
We have $m$ users and $n$ items.\footnote{We use item/question and label/option/answer/choice 
interchangeably.}
Each item has up to $k$ labels,\footnote{In other words, the item(s) with the most unique labels has $k$ different labels. Labels can be proposed either from questioners or answerers.} and the labels usually vary between items (the items are thus ``heterogeneous'').
Each user chooses up to one label to each of the items,
and two users may choose the same label to the same item.
Our goal is to derive a principled way for \emph{determining a ranking of the users in terms of their ability to pick correct labels for the items}
based solely on the user responses.

\introparagraph{Connection to truth discovery}
Ability discovery can be considered a dual problem of the widely studied \emph{truth discovery} problem~\cite{Zheng:2017:TIC:3055540.3055547}.
The setup is similar; 
the difference is that the truth discovery problem measures success in finding the truth 
(thus the correct label for each item)
whereas our problem focuses on finding the correct ranking of the users by their relative abilities.
While the truth discovery problem occurs in a wide range of problems related to crowdsourcing and has been of intense focus
for the data management community~\cite{Zheng:2017:TIC:3055540.3055547}, the ability discovery problem gets little attention and is usually treated as a sub-problem:
if one knows the truth, it is easier to rank the users based on the choices they make.
In turn, if one knows the order of user abilities, it is easier to determine the truth.
However, we show in \cref{sec:exp} that 
it is not straightforward to rank users correctly, 
even when given the correct answers to the questions beforehand, 
which means even the perfect truth discovery method is not guaranteed to perform well for the ability discovery problem.
Furthermore, we argue that ability discovery is much more than a sub-problem of truth discovery and highlight its importance in two aspects:
1)~It has different application scenarios as we discussed in the examples. 
2)~Different from correctness of answers, user abilities are \emph{abstract and cannot to be measured directly}, which makes ability discovery results valuable but also hard to evaluate with \emph{no acknowledged ground truth}.

\introparagraph{Assumptions}
Similar as in truth discovery, 
we assume an objective total order on the labels of each item based on their correctness,
and a total order on the users based on their latent one-dimensional abilities for choosing the correct labels.

\introparagraph{Our approach}
We first define an idealized scenario in which the user responses are \emph{consistent} with their abilities 
across the items and characterize 
it as the response matrix having the \emph{Consecutive Ones Property} (C1P).
We then suggest an efficient spectral method 
that we call \HND (\HnD)
for reconstructing such ideal orderings if they exist.
We prove that in the ideal error-free scenario (better users always make better choices) 
our method is guaranteed to find the correct ranking,
which puts our approach on a stronger theoretical footing as a cross between a heuristic and an exact algorithm. 
Importantly, our method generalizes to the general non-ideal case 
and allows us to compare it with existing truth discovery methods for ranking users.
One key innovation in our work is the use of Item Response Theory (IRT) \cite{LH1997}
to model both label rankings as well as the propensity of users of different abilities to choose such labels,
and the previously not made connection to C1P.
We utilize 3 different generative models from the IRT literature to generate realistic synthetic data.
Experiments on these data show that
($i$) our new method is \emph{more accurate than existing truth discovery methods}, and 
($ii$)~it can also serve as scalable approach that \emph{reconstructs the C1P order if it exists} and \emph{generalizes much better on non-ideal inputs} than the only other C1P order reconstruction method that works in the non-ideal scenario.

\introparagraph{Contributions}
(1) We connect the notion of \emph{consistent responses} in heterogeneous multiclass classification 
to the well-known \emph{Consecutive Ones Property} (\cons) from seriation theory~\cite{pjm/FulkersonGross:1965,kendall, ABH}.
We argue that any principled solution to ranking users by their abilities 
should be able to recover the correct ranking when responses are consistent with 
abilities.

(2) We propose a novel yet simple 
adaptation of the \HITS algorithm~\cite{hits} that we call \HND (\HnD) 
for ranking users based on their abilities. 
We prove the surprising result
that \HnD \emph{recovers the consecutive ones ranking} of users 
when a unique such order exists.
Unlike fast combinatorial algorithms for finding the \cons ordering if it exists, 
\HnD \emph{can also deal with the general case when no such order exists}.
This makes \HnD an ideal candidate for our problem (and even becomes an exact algorithm in special settings).
We compare \HnD against \ABH~\cite{ABH}, which is the only other spectral approach that has these properties, and 
give intuitive and experimental evidence for why \HnD performs better.

(3) We show how 
\emph{Item Response Theory (IRT)}~\cite{LH1997},
which is widely deployed in educational testing,
provides 
a natural and mathematically principled theory 
(including generative models)
for modeling heterogeneous item ranking
that includes the C1P as a special case of consistent responses.

(4) We conduct extensive experiments on synthetic datasets generated by 3 polytomous IRT models 
and show that \HnD can outperform other existing truth discovery approaches in terms of accuracy of the user ranking. 
We also show (both in theory and with experiments) that \HnD has better scalability 
and accuracy than 
\ABH (which is the only other C1P reconstruction approach known today that can be used for the general ranking problem).

\introparagraph{Outline}
\Cref{sec:prob_definition} defines our problem, draws the connection to the C1P property,
and introduces IRT as generalization of C1P.
\Cref{sec:algorithm} introduces our approach, gives its formal properties and compares it to closely related work.
\Cref{sec:exp} presents experiments.
\Cref{sec:stateofart} discusses additional related work on truth discovery 
before \Cref{sec:conclusion} concludes.
\iflabelexists{sec:nomenclature}
{Our code is available online~\cite{HITSnDIFFS-code:2023}.}
{Code, proofs and more experiments are available online~\cite{HITSnDIFFS-arxiv:2023,HITSnDIFFS-code:2023}.
}

\section{Formal setup}
\label{sec:prob_definition}

\subsection{Ability discovery problem formulation}
\label{sec:formulation}
Consider a setting with $m$ users 
choosing among $k$ options for each of $n$ items.
Items are \emph{heterogeneous} in that 
they have \emph{different options},
as is the case in 
MCQs
used in standardized test settings (see \cref{fig:Fig_Intro}a).
This setup is different from typical
multiclass classification~\cite{dawid} where all $n$ items 
have the same class of $k$ labels.
To emphasize the difference, we refer to our setup as \emph{heterogeneous multiclass classification}.

User responses can be represented in one-hot encoding as a ($m \times kn$) \emph{binary response matrix} 
$\response$ (see \cref{fig:Fig_Intro}b)
where each row represents the choices of a user
and each column represents an option for some item.
The number of non-zeros in $\response$ is $mn$ and the number of non-zeros in each row $n$.

Our goal is to rank the users by their abilities to choose the ``best'' options for each item.
Each user $\user_j$ has a latent one-dimensional \emph{selection ability} $\theta_j$ that represents the user's ability to choose the best of options for each item.

\begin{definition}[Ability discovery]
	Given $m$ users and their choices among $k$ options\footnote{To simplify the discussion and different from \cref{sec:intro}, we assume here that each item has exactly $k$ choices. For items with $k'\!<\!k$ choices, we can assume them to have $k \!-\! k'$ more choices and nobody choosing them.}
	for $n$ \emph{heterogeneous} items
	as binary response matrix $\response$.
	Rank the $m$ users by their abilities to correctly choose options for each item.
\end{definition}	

Several approaches on homogeneous items assume the probabilities of users getting a correct answer to be identical across questions and encode user abilities as a $(k \times k)$ dimensional confusion matrix per user \cite{dawid,Zheng:2017:TIC:3055540.3055547}.
In our setting, this is not the case: 
each option $h$ for item $t_i$ may have a \emph{different} 
``discrimination score'' $\alpha_{ih}$.
In general, the higher $\alpha_{ih}$ for option $h$ is, the more ``discriminating'' it is 
(the option's probability of being chosen more quickly increases with student ability).
The probability of answering a question correctly is then
a function of the user ability and all the question option discrimination scores. 
This setup builds upon \emph{Item Response Theory (IRT)}~\cite{LordNovickBirnbaum:1968:StatisticalTheories,LH1997}, summarized in detail in \cref{sec:irt}.

\begin{example}
\label{example:fig1}
\Cref{fig:Fig_Intro}a shows $m\!=\!4$ users answering $n\!=\!3$ MCQs. 
Each question has $k\!=\!3$ choices labeled A to C in decreasing order of fit.
\Cref{fig:Fig_Intro}b shows an example response matrix $\vec C'$ and its binary form $\vec C$.
Assuming that users' choices are ``consistent'' with their abilities 
(i.e.\ correctness of labeling increases with ability), the only possible ranking of users for the observed $\vec C'$ is 1, 2, 3, 4, or its reverse.
\Cref{fig:Fig_Intro}c illustrates an IRT model between the latent user abilities and the probability of picking the correct answer when the correct ranking is 1, 2, 3, 4.
For example, user $u_2$ has the ability to label items $t_1$ and $t_2$ correctly, thus user $u_2$ chooses the correct answer A for both items.
Our problem is to rank the users by their mastery of the subject based only on the users' choices without knowing the true labels.
\end{example}

\begin{figure}[t]
\centering
\includegraphics[scale=0.45]{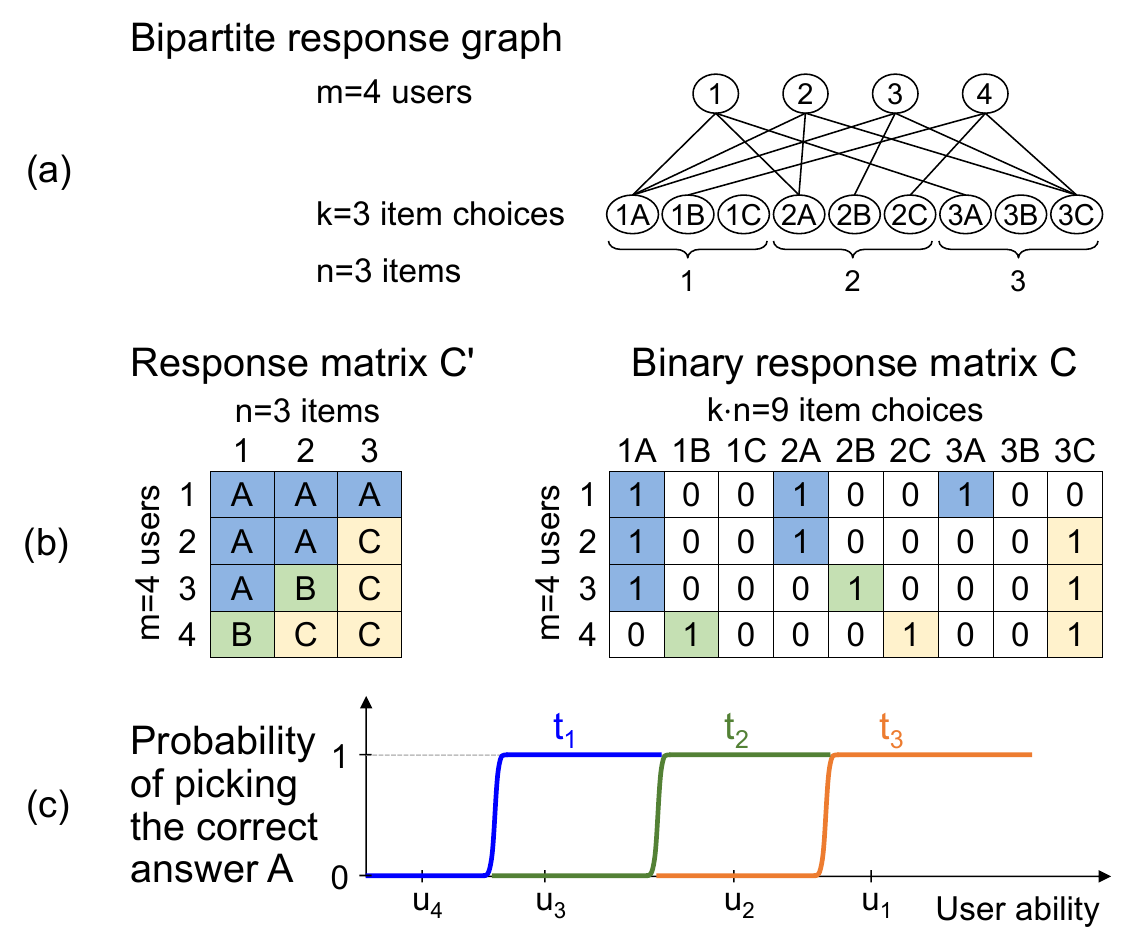}
\caption{(a) Ability discovery problem: $m\!=\!4$ users choose one from $k\!=\!3$ choices of labels A, B, C for each of $n=3$ items.
(b) Input: the $(m \times k)$ response matrix $\vec C'$,
or equally its flattened $(m \times k n)$ binary response matrix $\response$.
(c) Model: the probability of picking the correct answer in terms of the user ability for Items 1,2,3. The abilities of all 4 users are marked on the horizontal axis.
}
\label{fig:Fig_Intro}
\end{figure}

\subsection{The ideal case with consistent responses}
\label{sec:ideal_case}

We call a response matrix to be ``\emph{consistent}''
if there is an unambiguous ordering of users according to their abilities that is reflected in their responses across the items.
This way, the ability is a unique skill that is tested across their responses to all items.
In this ideal case, \emph{if a user $\user_1$ chooses a better option for an item $t_1$ than user $\user_2$,
then $\user_1$ must also choose an equal or better option for any other item $t_2$}
to reflect this consistency.
This implies that there is an implicit ordering among the choices for each item from best to worst and the better users pick better options for every item.

More formally, assume that the user abilities $\theta_j$ are all distinct, and also that for every item $t_i$, the discrimination $\alpha_{ih}$ over all the options are distinct, then there is 
\emph{a unique linear ordering of the users, and of the options for each item}.

\begin{definition}[Consistent Responses]
A response matrix $\response$ is \emph{consistent} if
there exists an assignment of user abilities $\bm{\theta}$
and item discriminations $\bm{\alpha}$, s.t.\
for any pair of users $u_1$ and $u_2$ with $\theta_1 > \theta_2$, and for any item $t_i$ where $u_1$ chooses option $h_1$ and $u_2$ chooses option $h_2$, we have $\alpha_{ih_1} > \alpha_{ih_2}$.
\end{definition}

\subsection{Relation to Consecutive ones Property (C1P)}

We observe that consistent response matrices, when row-sorted according to user abilities, satisfy a widely studied ordering property in seriation, called the \emph{consecutive ones property} (\cons)~\cite{pjm/FulkersonGross:1965,kendall, ABH}. 
We follow the notation from seriation theory and call it a \emph{P-matrix}.

\begin{definition}[\cons, P-matrix \& pre-P-matrix~\cite{ABH}]
A binary matrix satisfies \cons 
and is called a \emph{P-matrix}
if in each column, all the $1$'s are consecutive.
If the rows of a matrix can be permuted so it becomes a P-matrix, we call it a \emph{pre-P-matrix}.
\end{definition}

In other words,
no $0$'s appear between any two $1$'s in a column in a P-matrix (see $\vec C$ in \cref{fig:Fig_Intro}b).
To see that consistent responses with users sorted by abilities $\theta$ give a P-matrix,
suppose for a contradiction that a column corresponding to an option for an item has two or more blocks of ones.
Then the users corresponding to the zeros in between these blocks will have chosen another option for which the quality is strictly higher or lower than that of this option since the option qualities are assumed to be distinct.
But this violates consistency since the rows are ordered by user ability.
\begin{observation}[Consistent Responses imply C1P]
	A response matrix $\response$ is consistent iff it is a pre-P-matrix.
\end{observation}	

Consequently, ranking the users in a scenario of consistent responses
corresponds to the problem of
determining a permutation of the rows of $\response$ so that the result obeys \cons.

\introparagraph{State-of-the-art on C1P}
\label{sec:c1p}
Booth and Leuker \cite{boothleuker} (``\textbf{\BL}'') proposed the fastest known algorithm 
for finding all possible permutations of the rows that reconstruct the \cons ordering 
in time linear in the number of nonzero entries in the matrix, taking time $O(mn)$ in our setting.
However, their method fails to produce an ordering of the rows when the matrix is not a pre-P-matrix, 
and therefore cannot be used as a general heuristic for simulated or real-world datasets that are not ideal.
In contrast, Atkins et al. \cite{ABH} (``\textbf{\ABH}'') proposed an elegant spectral method 
to determine whether a matrix obeys \cons, thus giving a rare continuous method to solve a combinatorial ordering problem.
Moreover, it can also be adapted as a heuristic when the input matrix is not a pre-P-matrix.
To the best of our knowledge, this is the only currently known method that can be used for 
ability discovery that is also guaranteed to recover a C1P ranking if it exists.

\introparagraph{Our goal}
Our goal is to develop a fast and principled algorithm 
that just like \ABH
($i$) returns a P-matrix in the special case of pre-P matrix inputs, and
($ii$) can solve the problem in the more general case when the response matrix is \emph{not pre-P}.
As we show, the performance of \ABH quickly drops in the non-ideal setting (the general IRT setting in \cref{sec:irt}) which makes it unusable for ability discovery.
We show that our method is more robust and generalizes better, thus being the first method with a useful accuracy for ability discovery that is also guaranteed to solve the ideal case.

\subsection{Relation to Item Response Theory (IRT)}
\label{sec:irt}
\introparagraph{Brief introduction of IRT models}
A large body of work from psychological and educational researchers called \emph{Item Response Theory (IRT)}~\cite{LordNovickBirnbaum:1968:StatisticalTheories,LH1997}
studies the mathematical functions relating the probability of an examinee's response on a test item to an underlying ability.
All major educational tests, such as the Scholastic Aptitude Test (SAT)~\cite{LordNovickBirnbaum:1968:StatisticalTheories} 
and Graduate Record Examination (GRE)~\cite{kingston1982feasibility}, are based on IRT.
IRT forms a mathematically principled, experimentally validated, and widely used theory on how users answer items.
\cref{fig:IRT_overview} summarizes 
the connections between different IRT models
and 
\iflabelexists{sec:nomenclature}
{\cref{app:IRT}
contains even more details on IRT.}
{our online appendix~\cite{HITSnDIFFS-arxiv:2023}
contains even more details on IRT.
}

Binary IRT models can be seen as variations of the standard \emph{logistic} or sigmoid function
$\sigma : \R \rightarrow [0,1]$ defined by
$\sigma(x) = \frac{e^x}{1+e^x} = \frac{1}{1+e^{-x}}$,
which is widely used in machine learning models and a smooth relaxation of the Heaviside step function $H(x) = \mathbb{I}(x>0)$~\cite{pml1Book}.
These models describe the probability $\P_{i}(\theta)$ for a student with ability $\theta$ to answer a question $i$ correctly.
Here we only show the 3PL IRT model due to limited space, where $\theta$ is the latent ability for each user; $a$, $b$, $c$ are latent item factors characterizing the questions and their options, representing discrimination, difficulty and random guessing respectively
(Other models can be derived based on \cref{fig:IRT_overview}):
\begin{align*}
	\P_{i}(\theta) = c_i + (1-c_i)\sigma\big(a_i(\theta-b_i)\big) = c_i + \frac{1 - c_i}{1 + e^{-a_i(\theta-b_i)}}	
\end{align*}

Multinomial models measure the probability $\P_{ih}(\theta)$ for a student with ability $\theta$ to choose an option $h$ for a question $i$.
Thus different from binary models whose parameters belong to questions, multinomial IRT models assign parameters to each option.
For example, the Graded Response Model (GRM)  model \cite{samejima1997graded} considers a difficulty parameter for each option and a discrimination parameter for each question:
\begin{align*}
	&\P_{ih}(\theta) = \P^*_{ih}(\theta) - \P^*_{i,h+1}(\theta) \\
	&\P^*_{ih}(\theta) =  \sigma\big(a_i(\theta-b_{ih})\big) = \frac{1}{1 + e^{-a_i(\theta - b_{ih})}}	\\ 
	&-\infty = b_{i0} < b_{i1} < ... < b_{i,k-1} < b_{ik} = \infty
\end{align*}
The Bock model \cite{bock1972estimating} furthermore assigns a discrimination parameter to each option, and the Samejima model \cite{samejima1979new} takes into account random guessing by adding a dummy option.

\begin{figure}[t]
\includegraphics[scale=0.23]{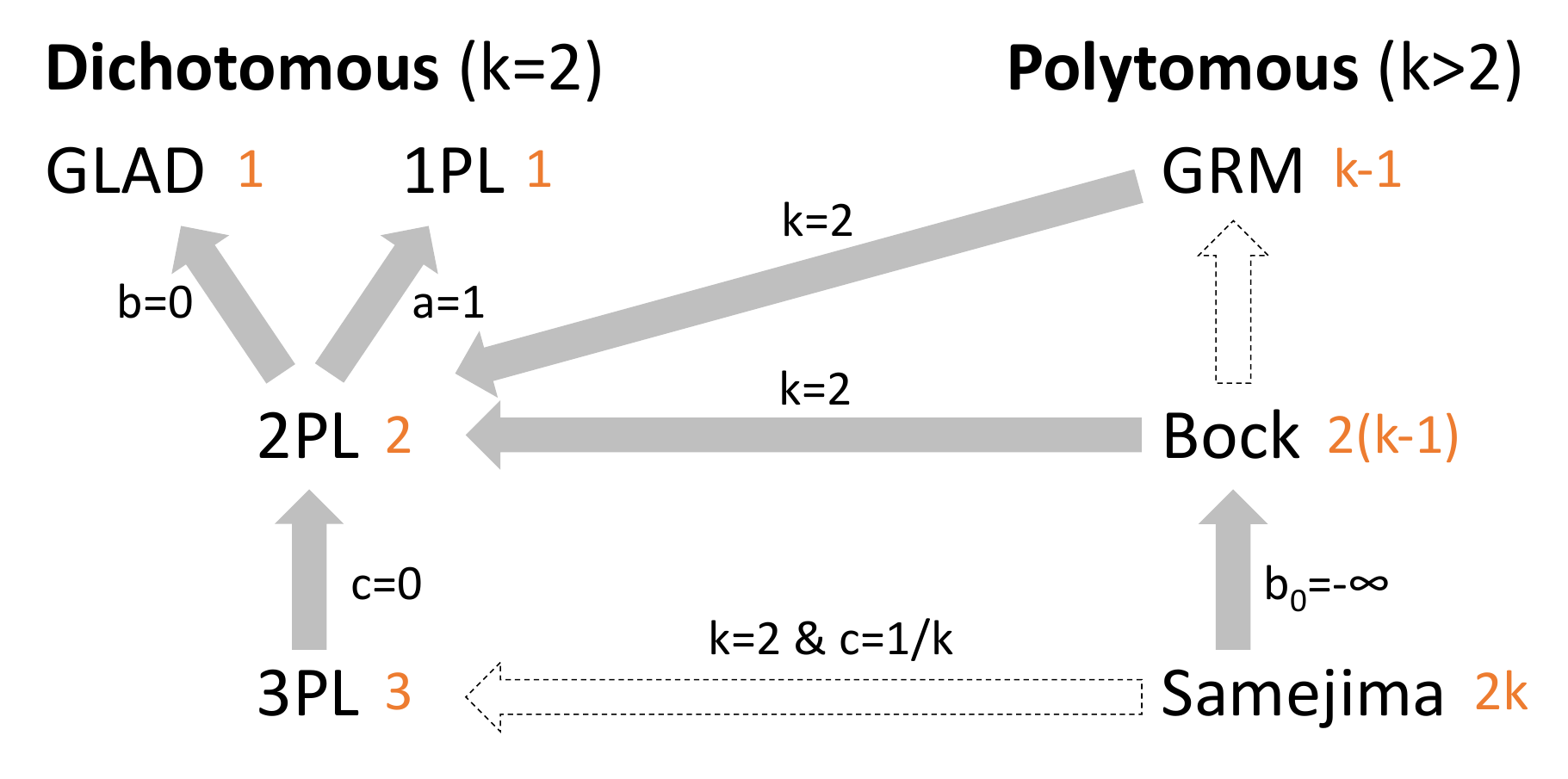}
\centering
\caption{Correspondences between the discussed IRT models.
Orange numbers show number of free parameters per question.
Arrows mean ``specializes into.'' 
Dashed arrows imply specialization requires special assumptions:
Bock to GRM: holds only approximately after fixing $a^{\textrm{Bock}}_{h} \!=\! h \cdot a^{\textrm{GRM}}$,
Samejima to 3PL: for $k \!=\! 2$ when $c \!=\! 1/k$.
}
\label{fig:IRT_overview}
\end{figure}

\introparagraph{Connection between IRT and C1P}
We introduced the definition of consistent responses in \cref{sec:ideal_case},
When a response matrix is consistent, it is a pre-P-matrix and can be permuted to become a P-matrix which has the consecutive ones property (C1P).
If the response matrix is a pre-P-matrix, 
the corresponding response function of the probability for a user to choose a specific option 
$h \in \{0, \ldots, k-1\}$ 
can be expressed as the difference between two Heaviside step functions:
\begin{align*}
	\P_{ih}(\theta) &= H(\theta - b_{ih}) - H(\theta - b_{i,h+1})
\end{align*}
for appropriately chosen $b_{ih}$ such that:
\begin{align*}
	&-\infty = b_{i0} < b_{i1} < ... < b_{i,k-1} < b_{ik} = \infty
\end{align*}
Notice that this response function is exactly the GRM model in the limit of $a \rightarrow \infty$:
Whenever the user ability $\theta$ is between $b_{ih}$ and $b_{i,h+1}$, then this student chooses option $h$.

To summarize, IRT models can be seen as a relaxed version of the response function in the ideal case when the response matrix can be permuted to obey C1P.
As widely accepted models for MCQs, they strongly support our principle of satisfying the more strict C1P in the ideal case and can be used to generate data in non-ideal cases (see \cref{sec:exp}).

\section{A family of HITS algorithms}
\label{sec:algorithm}

We review the \HITS algorithm and variants that have been proposed for truth discovery.
We then describe a natural averaging version of \HITS that we call ``\avgHITS.''
Our key observation is that the \emph{eigenvector corresponding to the 2nd largest eigenvalue} 
of its update matrix reconstructs the user (or row) ordering 
with \cons if one exists and is unique.\footnote{We consider an ordering and its reverse ordering to be the same.}
We then describe a variant that we call ``\HND'' (or \HnD in short) on a tripartite graph to find this eigenvector
efficiently:
it keeps an additional vector of differences between adjacent scores and updates it in the loop of the \avgHITS algorithm to compute the ordering of users that we require.
We prove that \HnD returns the correct ordering when the user responses are consistent
and compare its time complexity and expected accuracy for non-consistent responses with other methods.

\introparagraph{Required background from spectral graph theory}
We say that $\vec v$ is an eigenvector of $(n \times n)$ quadratic 
matrix $\vec A$ with eigenvalue $\lambda$ if
$\vec v \neq 0$ and
$\vec A \vec v = \lambda \vec v$.
If $\vec A$ is symmetric, then all eigenvalues are real~\cite{lay2016linear}.
We use indices to refer to eigenvalues sorted algebraically:
$\lambda_1  \geq \lambda_2 \geq \ldots \lambda_{n-1} \geq \lambda_{n}$.
We write $\vec v_i$ to refer to the eigenvector corresponding to eigenvalue $\lambda_i$
and will be cavalier in referring to it as ``the $i$-th largest eigenvector'' 
when we really mean ``the eigenvector corresponding to the $i$-th largest eigenvalue.''

If the matrix is non-negative and irreducible, then according to the Perron-Frobenius theorem~\cite{perron1907theorie, frobenius1912matrizen, Meyer00}
the first eigenvector is also the largest by amplitude 
($\lambda_1 = \max_i |\lambda_i|$)
and the corresponding eigenvector $\vec v_1$ is positive.

\subsection{``HITS'' and its variants for truth discovery}
\label{sec:hits}

\textbf{Hubs and Authorities (\HITS)}~\cite{hits}\footnote{HITS originally stood for ``Hyperlink-Induced Topic Search.''} is a classic spectral approach 
that several truth discovery approaches have built upon.
The original aim of the algorithm is to rate web pages by two scores: authority and hub score.
These scores are recursively defined such that the hub scores are proportional to the sum of the authority scores of the nodes they point to, and the authority scores are proportional to the sum of the scores of the hubs pointing to them, thus reflecting a mutually consistent set of scores~\cite{Newman:2010:NI:1809753}.

In the context of truth discovery, the authority and hub scores can be interpreted as the user abilities and option correctness scores, respectively. 
Let $\response$ represent the $m \times n$ binary response matrix
where $C_{j,i}=1$ iff user $j$ chooses item $i$, 
and $\userscore$ be the $m$-dimensional user score vector,
and $\qweight$ be the $n$-dimensional item weight vector.
In matrix notation, the scores are recursively defined as:
\begin{align*}
	\userscore \leftarrow \beta \response \qweight
	\qquad\qquad
	\qweight \leftarrow \alpha \response^\transpose \userscore	
	\quad
\end{align*}
where $\alpha$ and $\beta$ are normalization constants
and $\response^\transpose$ is the transpose of $\response$.
The algorithm starts from an initial assignment, such as $\userscore = \vec e$ and
iteratively updates then normalizes $\qweight$ and $\userscore$.
The user scores $\userscore$ will converge to the 1st eigenvector of the matrix $\response \response^\transpose$.

\textbf{TruthFinder}~\cite{DBLP:journals/tkde/YinHY08} modifies HITS 
by first taking \emph{the average instead of the sum} of the chosen item scores as user scores
and interpreting them as probabilities of the users being correct on any question.
It then defines an item's score 
as the probability of it being true given the independent probabilities of all the users choosing the item.
When appropriately initialized, the approach does not require normalization.
In the following matrix formulation, let $\responserow$ represent the row-normalized response matrix $\response$:
\begin{align*}
	\userscore \leftarrow \responserow \qweight
	\qquad\qquad
	\qweight \leftarrow \vec 1 -  
	\exp \big( \response^\transpose \log (\vec 1 - \userscore) \big)
	\quad
\end{align*}

\textbf{Investment} \cite{PasternackR2010:FactFinder} 
calculates the ability of users as the sum of the scores of their chosen options, weighted by the user ability they invested in the previous iteration.
\textbf{PooledInvestment} \cite{PasternackR2010:FactFinder} extends Investment by using a different formula for the item scores.
Both variants use non-linear scaling of the item scores with different user-specified hyperparameters.

\introparagraph{Our method}
We build upon this key idea of updating scores in a bipartite graph 
of users and items
by iteratively summing scores from one side to update the other.
However, in contrast to other methods, we focus on the 2nd largest instead of the dominant eigenvector
of a new variant and show that this approach is \emph{guaranteed to recover the correct ranking in case of consistent responses}.
As we will show in \cref{sec:exp}, no other existing truth discovery method can do that.

\subsection{``\avgHITS''}
\label{sec:avghits}

In our setup, there are $nk$ different choices ($k$ choices for each of $n$ items).
Consider a bipartite graph $G = (L \cup R,E)$ that corresponds to the ($m \times nk$) response matrix $\response$:
Partition $L$ contains a vertex for each of $m$ users: $L = \{\user_1,...,\user_m\}$.
Partition $R$ is a collection of $n$ vertex sets $R = \{I_1,...,I_n\}$, one for each item.
Each set $I_i$ contains $k_i \leq k$ vertices $I_i = \{c_{i1},...,c_{ik_i}\}$ where $c_{ih}$ represents option $h$ of item $i$.
We add an edge to between a user $\user_j$ and an option $c_{ih}$ $E$ if user $j$ chooses option $h$ for item $i$.

To make our derivations easier to follow, we will conveniently assume that each item $i$ has the same number $k_i = k$ of options. Notice however, that this is not required for our approach.
We further assume $\response$ to be connected. This requirement applies to all spectral truth ranking methods including \HITS as otherwise the relative ranking of users (or items) from different components can't be established.\footnote{PageRank achieves the connectivity with the teleport operation.}
Finally, define $\userscore$ as a ($m \times 1$) user score vector
and $\qweight$ as a ($kn \times 1$) option weight vector denoting weights for each of the $kn$ options
	according to their order in $\response$.

We call \avgHITS the modification of the \HITS update rule 
that uses \emph{averages} instead of sums
to iteratively update the user scores and option weights \emph{in both directions}:
the user score $s_j$ is updated to be the average of the weights of all the options that user $j$ picked,
and an option weight $c_{ih}$ is updated to be the average of the scores of all users who picked it.
In the following matrix formulation, 
let $\responserow$ represent the row-normalized
and $\responsecol$  the column-normalized response matrix $\response$.
At each iteration (until convergence), we update 
the user score vector $\userscore$ 
and the option weight vector $\qweight$ as follows:
\begin{align*}
	\userscore \leftarrow \responserow \qweight
	\qquad\qquad
	\qweight \leftarrow (\responsecol)^\transpose \userscore
	\quad
\end{align*}

By combining the
above two update equations,
we can update user scores between iterations directly
by replacing the two normalized response matrices with one \emph{update matrix} $\update$:
\begin{align}
	\textbf{s} &\leftarrow \underbrace{\responserow(\responsecol)^\transpose}_{\update} \textbf{s}
\end{align}

These iterations are not yet very helpful. Indeed, we observe 
that the largest eigenvector of $\update$ is the all-ones vector $\onesvec$,
and this is the vector of user scores that $\avgHITS$ converges to.
It turns out that it is the eigenvector corresponding to the \emph{2nd largest eigenvalue} of $\update$ that we seek. 

\subsection{Our algorithm ``\HND'' (\HnD)}
\label{sec:HND_algo}

In the following, we show a simple algorithm to find the 2nd largest eigenvector ordering of $\update$ and prove that 
it can be used to find the unique consecutive ones ordering of the response matrix $\response$.
By ``the eigenvector ordering'', we mean the ranking of entries in this eigenvector in terms of their values.
For example, 
$\vec v_1 = \{0.36, 0.8, 0.48\}$ 
and
$\vec v_2 = \{0.48, 0.64, 0.6\}$
have the identical ordering $\{3, 1, 2\}$ or its reverse $\{1, 3, 2\}$.
Our algorithm does not return the 2nd largest eigenvector of $\update$ but instead returns a vector with the \emph{identical ordering}.

The 2nd largest eigenvector of a matrix can be found using a variant of the deflation method \cite{pml1Book, DBLP:conf/nips/Mackey08}, 
which we will discuss
in detail
in \Cref{sec:complexity}.
Here we present a novel, conceptually simple, and faster algorithm 
that we term ``HITS and DIFFS'' (\HND or \HnD) 
that extends
\avgHITS from a bipartite to a tripartite graph 
and whose iterative updates converge to a user ranking that is
\emph{guaranteed to be \cons in the ideal case}, 
and that \emph{performs well also in more general settings.}
This approach leverages particular properties of our problem that don't apply 
to the 2nd largest eigenvector orderings of any matrix
from more general settings.

First, we propose a new intermediate step that calculates differences between user scores in the iterative updates of \avgHITS.
Rather than updating the user scores iteratively,
\HND updates the \emph{differences} between adjacent user scores by using a suitably modified update matrix,
and results in the scores converging to the ordering according to the second largest eigenvector of $\update$.
Furthermore, this modification only adds a linear overhead of computing the user score difference vectors and normalizing it in every iteration.
As we will show in \cref{theorem: main}, 
when the response matrix obeys \cons, then \HnD reconstructs the correct ordering of the rows.

\begin{figure}[t]
\captionsetup[subfigure]{justification=centering}
	\centering
	\includegraphics[scale=0.5]{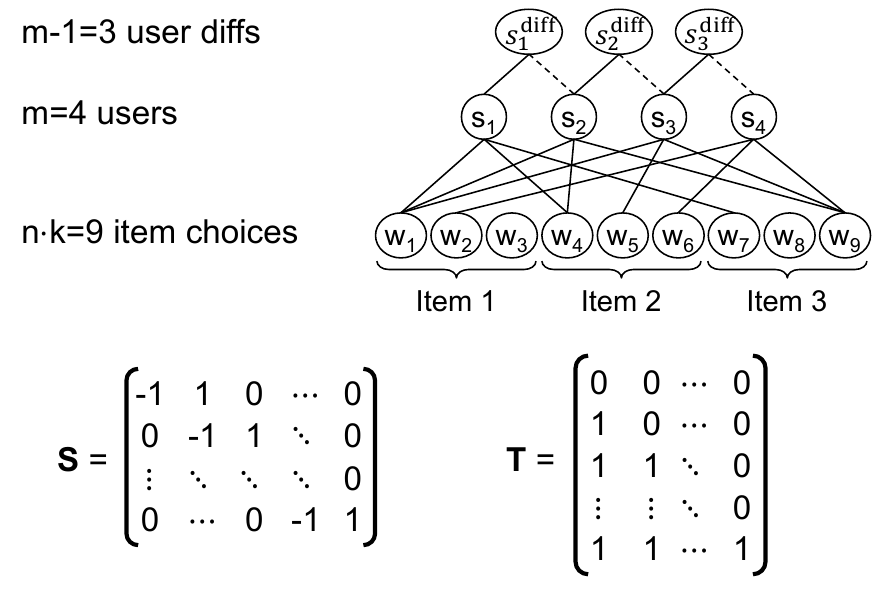}
\caption{\HND uses a \emph{3-partite graph} of option weights, user scores, and user diffs.
Contrast this graph with \cref{fig:Fig_Intro}.
The update equations 
(see~\cref{alg:hnd})
use two re-shaping matrices $\vec S$ and $\vec T$.}
\label{fig:Fig_Tripartite}		
\end{figure}

As shown in \cref{fig:Fig_Tripartite}, 
we define a new vector $\userdiff$ of differences in user scores with entries
$\singleuserdiff_j = s_{j+1} \!-\! s_{j}$, $j \in [m\!-\!1]$.
This is equal to 
$\userdiff_j =  \vec{S} \userscore$
where $\vec{S} \in \Rset^{(m-1) \times m}$ is 
shown in \cref{fig:Fig_Tripartite}.
In the reverse direction, there are infinitely many vectors $\userscore$ that can be 
generated from a given $\userdiff$, all shifted by different constants.
Since we only want a final ordering of users,
we can WLOG set the first element of the vector $\userscore$ to be $0$.
The transformation then is
$\userscore = \vec{T} \userdiff$
where $\vec T \in \Rset^{m \times (m-1)}$ is the lower unit triangular matrix\footnote{It is this fixing that intuitively keeps the ordering, but changes the actual amplitudes.}.

We can now get a user difference score update rule:
\begin{align}
	\vec{s}^{\textrm{diff}} 
	&\leftarrow \vec{S} \userscore 
	= \vec{S} \responserow (\responsecol)^\transpose \userscore 
	= \underbrace{\vec{S} \responserow (\responsecol)^\transpose \vec{T}}_{\updatediff} \vec{s}^{\textrm{diff}}
	\label{eq:hitsndiffs_naive}
\end{align}

\noindent
In other words,
$\updatediff = \vec{S} \update \vec{T}$
is a ``difference update'' matrix that is used to update $\userdiff$ from one iteration to the next.
With these update equations, $\vec{s}^{\textrm{diff}}$ converges to the largest eigenvector
of $\updatediff$.
Our algorithm \HnD that implements this is described in \cref{alg:hnd}.

We now prove the connection between the 1st eigenvector of $\updatediff$
and the 2nd largest eigenvector of $\update$.

\begin{lemma}[Eigenvector correspondence]
\label{lemma:EigenvectorCorrespondence}
$\vec{x}$ is the 2nd largest eigenvector of $\update$
iff
$\vec{y} = \vec{S} \vec{x}$ is the largest eigenvector of $\updatediff$.
\end{lemma}

\introparagraph{Proof sketch}
Due to the limit of space, we only provide the high-level ideas of our proofs in the paper.
First, we can find out that each row of $\update$ has sum 1.
Using this, we can prove that the largest eigenvector of $\update$ is in the direction of the all ones vector $\vec{e} = \vec{1}_m$ if the largest eigenvalue of $\update$ has multiplicity 1 
(i.e.\ the graph is a single connected component).
Let $\vec{x}$ be an eigenvector of $\update$ that is not in the direction of the all ones vector, i.e. $\vec{x} \neq \alpha \vec{e}$.
Note that $\vec{T} \vec{S} = (\vec{I}_m - \vec{e}\vec{e}_1^\transpose)$ and each row of $\vec{S} \update$ sums to $0$.
Then,
\begin{align}
\update \vec{x} &= \lambda \vec{x}  				\notag \\
\vec{S} \update \vec{x} &= \lambda \vec{S} \vec{x}  \notag \\
\vec{S} \update (\vec{I}_m - \vec{e}\vec{e}_1^\transpose) \vec{x} &= \lambda \vec{S} \vec{x}  \notag \\
\vec{S} \update \vec{T} \vec{S} \vec{x} &= \lambda \vec{S} \vec{x}  \notag \\
\updatediff \vec{y} &= \lambda \vec{y}, \quad \textrm{where} \quad \vec{y} = \vec{S} \vec{x} 	
\label{eq:s to s_diff}
\end{align}

Therefore, $\updatediff$ has exactly the same eigenvalues as $\update$ except the largest eigenvalue $1$, and the eigenvectors of $\updatediff$ are the differences between the entries of the corresponding eigenvector of $\update$.
Thus we prove the lemma.
\qedsymbol{}

\begin{theorem}[2nd eigenvector of \avgHITS recovers C1P]
\label{theorem: unique}
If $\response$ is a pre-P-matrix with a unique consecutive ones ordering of its rows and each row has the same row sum, then this ordering of the rows of $\response$ is given by the ranking of the rows sorted by values in the 2nd largest eigenvector of $\update$.
\end{theorem}

\introparagraph{Proof sketch}
We can first prove that if $\response$ is a pre-P-matrix with a unique consecutive ones ordering of its rows and each row has the same row sum, $\update$ is an R-matrix (defined in \cite{ABH}) where in each row and column, the entries closer to the diagonal are larger than or equal to the further entries.
Using this, we can prove every entry in $\updatediff$ is non-negative by computing each entry in $\updatediff$ step by step according to its definition, which means $\updatediff$ is a non-negative matrix.
We can now apply the Perron-Frobenius Theorem \cite{perron1907theorie, frobenius1912matrizen}: 
there exists a non-negative eigenvector of $\updatediff$ corresponding to the largest eigenvalue of $\updatediff$. 
We know $\updatediff$ has exactly the same eigenvalues as $\update$, 
except the largest eigenvalue $1$, and the eigenvectors of $\updatediff$ are 
the differences between the elements of the corresponding eigenvector of $\update$.
Since the differences between the elements of the eigenvector corresponding to the 2nd largest eigenvalue of $\update$ (largest eigenvalue of $\updatediff$) is non-negative, 
that eigenvector of $\update$ is monotonic.
Therefore, sorting the rows according to the second largest eigenvector ordering of the corresponding update matrix $\update$ gives a P-matrix, proving the theorem.
\qedsymbol{}

\begin{algorithm}[t]
\caption{\HND (\HnD-power): 
A fast implementation of equation \cref{eq:hitsndiffs_naive} to calculate the 2nd eigenvector ordering of $\update=\responserow(\responsecol)^\transpose$
}
\label{alg:hnd}
\SetKwFor{ForAll}{forall}{do}{endfch}	
\KwIn{Response matrix $\response$, randomly initialized user scores $\userscore_0$}
\KwOut{User scores $\userscore$}
\BlankLine
$\vec{s}^{\textrm{diff}} \leftarrow \vec{s}^{\textrm{diff}}_0$ 
\hspace{13.8mm}// initialize user score differences \;
\Repeat{convergence or iteration limit}
{
	$\vec s  \leftarrow \vec T \vec{s}^{\textrm{diff}} $	\label{alg:3} 
	\hspace{9mm}// update user scores \;
	$\qweight \leftarrow (\responsecol)^\transpose  \userscore$	\label{alg:4}	
	\hspace{3.4mm}// update option weights \;
	$\userscore \leftarrow \responserow  \vec \qweight$	\label{alg:5}		
	\hspace{7.4mm}// update user scores \;
	$\vec{s}^{\textrm{diff}}  \leftarrow \vec S \userscore$ \label{alg:6}	
	\hspace{9.5mm}// update user score differences \;
	Normalize $\vec{s}^{\textrm{diff}}$ to be a unit vector\;
}
$\vec s  \leftarrow \vec T \vec{s}^{\textrm{diff}} $
\end{algorithm}

\begin{theorem}
\label{theorem: main}
If $\response$ is a pre-P-matrix with a unique C1P ordering of its rows and each row has the same row sum, 
then \HnD reconstructs the consistent ordering of the users 
taking linear time in the number of nonzeros in $\update$ per iteration.
\end{theorem}

\introparagraph{Proof sketch}
From \cref{lemma:EigenvectorCorrespondence}, we know by converting the converged largest eigenvector of $\updatediff$ back into a user score, we regain the ordering of the rows according to values in the second largest eigenvector of $\update$.
This, along with \cref{theorem: unique}, proves this theorem that \HnD detailed in \cref{alg:hnd} reconstructs the ideal consistent ordering.
\qedsymbol{}

\subsection{Decile entropy-based symmetry breaking}
Notice that reversing the order of a P-matrix still leaves it as a P-matrix.
Thus all methods for solving \cons suffer from a natural \emph{symmetry breaking problem}:
they have to decide between the order returned by an algorithm \emph{or its exact inverse}.

Our solution to this symmetry-breaking problem is motivated by the following observation:
\emph{users with higher ability tend to converge on the correct option as a majority answer},
while users with lower ability at the other end of the ordering tend to answer randomly.
This idea is similar to the main argument in \cite{li2017hyper} 
that experts tend to answer similar correct answers.
Thus the lower end of the user ordering 
has a \emph{higher entropy} in the choices picked than the higher quality end.
Notice that this idea is also implicit in IRT models with random guessing
where users with low ability choose uniformly random among the options (hence have high entropy),
whereas users of high ability pick the single correct choice.

We operationalize this idea in a new heuristic that is very effective in practice:
Given a ranking of the users, we compute,
for the top and the bottom user decile,
the average entropy of the chosen item options across all items.
We pick the side with lower entropy as the users with higher quality. 
We use this ``\emph{decile entropy method}'' for both \HnD and \ABH in our experiments.

\subsection{Why \HnD works better than ABH}
\label{sec:Hnd_vs_ABH_theory}

\HnD and \ABH rely on strikingly similar intuitions about spectral properties of matrices:
\HnD ranks users by the 2nd largest eigenvector of $\update$ 
(whose difference is the largest eigenvector of $\updatediff = \vec{S} \responserow (\responsecol)^\transpose \vec{T}$),
whereas 
it can be shown that \ABH ranks users
by the 2nd smallest eigenvector of the Laplacian matrix $\laplacian$ of $\response\response^\transpose$
(whose difference is the smallest eigenvector of $\vec M = \vec S \laplacian \vec T$).
In an ideal scenario with consistent responses 
(thus in an IRT scenario with very large discrimination scores), 
both methods are guaranteed to return the correct C1P ordering.\footnote{Recall that they are guaranteed to return the same ordering, but not the same eigenvector.}
	
We can also expect the accuracy to be identical in the other extreme scenario where all questions have 0 discrimination. 
But how can we expect their accuracy to compare in the more general scenario?

We interpret the general IRT scenario as random perturbations~\cite{stewart1990matrix} from the ideal C1P case.
Now notice that the smallest eigenvector of $\vec M$ 
is identical to the largest eigenvector of 
$\beta \vec I_{m-1} - \vec M$ where $\beta$ is larger than all the entries and all the eigenvalues of $\vec M$.\footnote{To see that, assume $\vec A \vec v= \lambda \vec v$.
Then $(\vec A + \beta \vec I) \vec v = \vec A \vec v + \beta \vec I \vec v = \lambda \vec v + \beta \vec v = (\lambda + \beta) \vec v$.
Thus if $\vec v$ is an eigenvector of $\vec A$ with eigenvalue $\lambda$, then 
$\vec v$ is also an eigenvector of the spectrally shifted matrix $\vec A + \beta \vec I$, but with eigenvalue $\lambda + \beta$.}
Thus the comparison of \HnD and \ABH corresponds to the largest eigenvector of $\updatediff$ against $\beta \vec I - \vec M$.
Since both matrices have all non-negative entries in the ideal scenario,
we know from the Perron-Frobenius theorem \cite{perron1907theorie, frobenius1912matrizen} 
that all values in their largest eigenvector are non-negative.

The user score of the $k$th user equals to the cumulative sum of the first $k-1$ entries in the eigenvector. 
In the ideal case when $C$ is a C1P matrix, every entry of $\vec{s}^{\textrm{diff}}$ is non-negative so $\userscore$ give us a perfect ranking of the students. 
In the non-ideal scenario when $C$ is not a C1P matrix and the users are permuted by their abilities, 
the sign of the entries in $\vec{s}^{\textrm{diff}}$ change.
When the eigenvector is even, a simple sign change in one of the entries does not influence the entire ranking of $\userscore$ but when the eigenvector has a large variance, a simple sign change in one large entry can break the entire ranking.
For example, if the $k$th entry of $\vec{s}^{\textrm{diff}}$ is quite large but the sign is negative, the error of the ranking of the ($k+1$)th students can be very large. 

Based on our observation above, we expect \HnD to work better than \ABH as the variance of largest eigenvector of $\beta \vec I - \vec M$ should be much larger compared to $\updatediff$.
The result is verified with dedicated experiments in \cref{sec:stability}:
\Cref{fig:comparison_variance} shows our observations on the variances of the $\vec{s}^{\textrm{diff}}$ for $\beta \vec I - \vec M$ and $\updatediff$.
\cref{fig:comparison_difference} and \cref{fig:comparison_accuracy} verify that \ABH is less accurate and less stable than \HnD.

\subsection{Complexity Comparison}
\label{sec:complexity}

We analyze the asymptotic time complexity of the various methods.
We compare \HnD against ($i$) existing C1P reconstruction algorithms \BL and \ABH, 
and ($ii$) the deflation method~\cite{pml1Book, DBLP:conf/nips/Mackey08} as an alternative method to compute the 2nd largest eigenvector for \avgHITS.

The time complexity depends on the number of users $m$, 
the number of questions $n$, and the number of iterations $t$ which may be different for different methods.
	
We assume $t \ll n$ and $t \ll m$ and thus only focus on $m$ and $n$.
Notice that although the response matrix $\response$ is a ($m \times kn$)-matrix, 
it has only $\O(m  n)$ 
non-zero entries since every user can pick only one label per question.

\introparagraph{\HND} 
A naive way to calculate the ranking is to first compute $\updatediff$ and then use the power method on it.
However, computing $\updatediff$ requires a matrix-matrix multiplications before the iterations with time complexity  $\bigO(m^2n)$.
Since $\updatediff$ is a $(m-1) \times (m-1)$ matrix,
the time complexity to run the power method on $\updatediff$ 
is $\bigO(m^2)$ per iteration.
This gives a total time complexity for the naive implementation as $\bigO(m^2n) + \bigO(m^2t) = \bigO(m^2n)$.
By instead running the mutual updates of $\qweight$, $\userscore$ and $\userdiff$ as described in \cref{alg:hnd}, 
we can replace matrix-matrix multiplications with several matrix-vector multiplications
and thereby get a more efficient implementation of \HND in \fbox{$\bigO(mnt)$}.
In other words, the speed-up results from applying the associativity law and replacing the calculation
$\vec s \leftarrow (\vec{S} \responserow (\responsecol)^\transpose \vec{T})
\vec{s}^{\textrm{diff}}$ 
with 
$\vec s \leftarrow \vec{S} (\responserow ((\responsecol)^\transpose (\vec{T}
\vec{s}^{\textrm{diff}})))$.
One detail is that we need however to implement \cref{alg:3} differently.
Materializing the matrix $\vec T$ would take $\bigO(m^2)$. 
Instead, we calculate the entries 
for $\vec s$ from 
$\vec{s}^{\textrm{diff}}$
via cumulative summation (e.g.\ via \texttt{numpy.cumsum} in Python).

\introparagraph{The deflation method}
\Cref{theorem: unique}
showed that our problem can be formulated as
finding the 2nd largest eigenvector $\vec v_2$ 
of $\update$.
The problem of finding the eigenvector corresponding to the second largest eigenvalue of a given matrix $\vec A$ can be solved with the \emph{deflation method}~\cite{pml1Book, DBLP:conf/nips/Mackey08}.
The idea is to first calculate the dominant eigenvector $\vec v_1$, and then eliminate the influence of the $\vec v_1$  from $\vec A$ to get a new matrix $\vec B$ whose 1st eigenvector is the 2nd eigenvector of $\vec A$.
Then $\vec v_2$ of $\vec A$  can be obtained by using the power method on $\vec B$.  
We next argue (and later show experimentally in \cref{sec:sca}) that using the deflation method is slightly less efficient than $\HnD$
(in addition to being not as simple to formulate as \HnD-power).

The most widely known deflation method 
\cite{pml1Book, DBLP:conf/nips/Mackey08} 

only works for symmetric matrices and does not apply to the asymmetric $\update$.
\cite{Paul58} presents several more variants of the deflation method including some that work for non-symmetric matrices.
Most of those methods either require matrix-matrix multiplication or both the left dominant eigenvector and the right dominant eigenvector.
The only exception is Wilkinson's vector annihilation \cite{Paul58Wilkinson}
which only needs the right dominant eigenvector 
(which we know is a unit vector in the direction of the all ones vector in our case). However, \cite{Paul58} claimed that this method is difficult to apply in practice because of the need to conduct annihilation between the power iterations and we found no open-source implementation.

For our experiments in \cref{sec:exp} we implement \emph{Hotelling's matrix deflation} \cite{Paul58Hotelling} which uses both the left and right largest eigenvectors and thus requires one more round of the power iteration.\footnote{
We first calculate the dominant left eigenvector via power iteration, then deflate the matrix, 
and then calculate the dominant right eigenvector on the deflated matrix.}
The experimental result in \cref{sec:sca} verifies that \HnD is not just conceptually simpler
but also slightly more efficient than the deflation method.

\introparagraph{ABH \cite{ABH}}
To reconstruct the C1P ordering, ABH requires the computation of the Fiedler vector \cite{fiedler1989laplacian}, 
which is the eigenvector corresponding to the 
2nd smallest eigenvalue of the related Laplacian matrix 
$\laplacian$.

The original \ABH paper \cite{ABH} does not propose an explicit solution and instead refers 
to the Lanczos algorithm \cite{lanczos1950iteration, lanczos}, whose time complexity is $\bigO(dmt)$ 
with $d$ being the average number of non-zero entries in a row of a given $(m \times m)$ matrix $\vec A$.
When the Laplacian matrix is dense as in our scenarios 
the time complexity of the Lanczos algorithm is $\bigO(m^2t)$.
It is efficient for eigenvector computations on large symmetric matrices
\cite{cullum2002lanczos},
and the state-of-the-art Fiedler vector solver~\cite{hu2003fast} uses Lanczos.

However, implementation by libraries such as Scipy~\cite{2020SciPy-NMeth} and Tenpy~\cite{tenpy}, require the full matrix as input, which would require us to compute the Laplacian matrix first. 
This calculation involves matrix-matrix multiplications and, as we show in \cref{sec:sca}, results in 
\fbox{$\bigO(m^2n)$}.

We provide another solution for \ABH which is not in the original paper \cite{ABH}.
Similar to how we implement \HnD as \cref{alg:hnd}, 
we can also implement \ABH by using the power method on the matrix $\beta \vec I_{m-1} - \vec M$ to get its largest eigenvector without having matrix-matrix multiplications.
As we discussed in \cref{sec:Hnd_vs_ABH_theory}, 
this largest eigenvector of $\beta \vec I_{m-1} - \vec M$ is identical to the smallest eigenvector of 
$\vec M$ which can be used to compute the order of the 2nd smallest eigenvector of $\vec L$ 
in the same way as \cref{alg:hnd}.
The total time complexity for this algorithm is \fbox{$\bigO(mnt + m^2t)$}.
When $m$ and $n$ are close or $m < n$, the time complexity becomes $\bigO(mnt)$ which is the same as \HnD.
However, when $m \gg n$, the time complexity becomes $\bigO(m^2t)$, which is larger than $\bigO(mnt)$ of \HnD.

\introparagraph{BL \cite{boothleuker}}
The original paper by Booth and Leuker (\BL) \cite{boothleuker} for reconstructing the \cons property 
can work directly on the initial response matrix 
and runs in $\bigO(mk + n + f)$, where $f$ is the non-zero entries in the matrix.
In our setup where $f = \bigO(mn)$, the time complexity is $\bigO(mn)$.
Therefore, BL is the fastest method when it works. 
Since it cannot be used for solving ability discovery in general, we are not using it in our experiments.

\introparagraph{\HITS \cite{hits}, Truthfinder \cite{DBLP:journals/tkde/YinHY08}, 
Investment \cite{PasternackR2010:FactFinder}, PooledInvestment \cite{PasternackR2010:FactFinder}}
All these methods are iteration-based variants of HITS 
that take at least $\bigO(mnt)$
and differ in how they iterate between the user and the item scores.
In practice, only \HITS can be defined as an eigenvector problem with a closed-form solution and efficient linear algebra implementation. Truthfinder converges in practice,
while Investment and PooledInvestment can not converge and return different results depending on initialization.
Our approach in contrast comes with the same computational properties as \HITS: 
guarantee of convergence, unique solution, an intuitive formulation as a spectral problem, and an efficient matrix implementation.

\section{Experiments}
\label{sec:exp}

Our experiments compare the \emph{accuracy} of ability discovery and the \emph{scalability} of the various methods.
The main takeaways from the experiments are:
1) \HnD robustly returns rankings for users with accuracy on average \emph{better than or equal to} other  
truth discovery methods;
2) \HnD is competitive with two ``\emph{cheating competitors}" (that are provided the ground truth information about the correct options for each question that is usually not available);
3) \HnD has better scalability as a C1P reconstruction algorithm than \ABH.

\subsection{Experimental setup}
\label{sec:benchmark}

\introparagraph{Environment}
All scalability experiments are run on Intel Xeon E5-2680 CPUs with an exclusive environment and 128G allocated memory.
We implemented \HnD in Python 3.8.1.

\introparagraph{Benchmarks}
\cite{sheshadri2013square} provides an extensive benchmark for the \emph{truth discovery problem}, 
and \cite{Zheng:2017:TIC:3055540.3055547} provides a broad survey of existing truth discovery methods.
Two points stand out: 
1) All open-source datasets used in the two papers \emph{lack ground truth for user abilities}.
2) All 20 datasets
fall under 
the
setting with \emph{homogeneous items}.
As we discussed in \cref{sec:intro}, it is easy to understand the lack of ground truth for the ability discovery problem because the \emph{user abilities are abstract and not able to be obtained from external knowledge}.
To make up for it,
we first create synthetic data based on 
the polytomous models from
\emph{Item Response Theory (IRT)} (recall \cref{sec:irt}).
Moreover, we use a real-world MCQ dataset with approximate (but not accurate enough) ground truth as a supplementary evidence to verify the usefulness of \HnD in \cref{sec:multichoice}.

\introparagraph{Polytomous synthetic data generator}
\label{sec:synth}
We use the three polytomous IRT models from~\cref{sec:irt} 
(GRM~\cite{samejima1997graded}, Bock~\cite{bock1972estimating} and Samejima~\cite{samejima1979new}) to generate synthetic data sets with known ground truth.
Samejima model takes random guessing into account so it models the educational test scenario where students try to maximize their scores.
Bock and GRM models with no random guessing models the crowdsourcing scenario where workers usually do not guess.

By default, we set user ability $\theta$ to be within $[0,1]$, 
item difficulty $b$ to be within $[-0.5,0.5]$, 
and the item discrimination $a$ to be within $[0,10]$, all uniformly random.
Besides varying the number of users, items and options, 
\cref{sec:acc} also has experiments with shifted $b$'s (chosen to achieve a certain percentage of users giving correct answers).\footnote{For experiments with GRM data, we use the data generator from the GIRTH package which requires at least $k=3$ options. 
We implemented Bock and Samejima generators ourselves and thus options can start from $k=2$.}

\introparagraph{Methods and their implementations}
For a thorough evaluation, we created three alternative implementations of \HnD and two alternative implementations of \ABH.
\textbf{\HnD-power} follows \cref{alg:hnd} which only involves matrix-vector multiplications. 
\textbf{\ABH-power} is our novel kimplementation of \ABH
that avoids matrix-matrix multiplications by using the power method on the matrix $\beta \vec I_{m-1} - \vec M$.
\textbf{\ABH-direct} is the implementation of \ABH using the Lanczos algorithm as suggested by \cite{ABH}.
\cref{sec:complexity} discussed the drawback of requiring matrix-matrix multiplications.
We use an efficient sparse Linear Algebra Python library called Scipy~\cite{2020SciPy-NMeth}.
\textbf{\HnD-direct} implements \HnD similarly by directly computing the 2nd largest eigenvector of $\update$ 
by using the Arnoldi algorithm \cite{arnoldi1951principle}, which can be considered the general version of the Lanczos algorithm on asymmetric matrices, also using Scipy.
\textbf{\HnD-deflation} implements \HnD with the deflation method discussed also in \cref{sec:complexity}.
For \HnD-power, \ABH-power and \HnD-deflation, the criterion for convergence 
is a maximal L2-norm of $10^{-5}$ over the change.
For our experiments (other than \cref{sec:sca}) we used ``\HnD-power'' for \HnD and ``\ABH-direct'' for \ABH 
since they turned out to be the fastest implementations.

We also implemented \textbf{HITS}~\cite{hits}, \textbf{TruthFinder}~\cite{DBLP:journals/tkde/YinHY08},
\textbf{Investment} and \textbf{PooledInvestment}~\cite{PasternackR2010:FactFinder}.
Since none of those iteration-based approaches (except HITS) allows an efficient matrix formulation, 
our implementation in Python uses loops and is not efficient.
We thus do not report scalability experiments on those methods as native implementation in C++ would bring those close to \HnD as discussed in \cref{sec:complexity}.
For Investment and PooledInvestment (which do not converge) 
we use 10 iterations instead of tuning the number of iterations.

\introparagraph{Two cheating baselines}
To show the effectiveness of \HnD, we also compare \HnD with two ``{cheating competitors}" that are given additional ground truth information about questions that is usually not available: 
\textbf{True-answer}
has information about which choices are correct for each question (which is usually not known in our scenario)
and then 
ranks users by the number of correctly answered questions.
\textbf{GRM-estimator} uses a Python package called GIRTH~\cite{girth}
that estimates the parameters of a GRM model including user abilities.
However, it \emph{requires knowing the order of options} for each question by correctness. 

Our comparison with this approach is notable because it is the theoretically ``best'' model 
to fit data generated by the same synthetic GRM process.

\subsection{Accuracy on synthetic data}
\label{sec:acc}

\begin{figure*}[h!]
\captionsetup[subfigure]{justification=centering}
    \centering
	\begin{subfigure}[h]{0.23\textwidth}
        \centering
        \includegraphics[width=\textwidth]{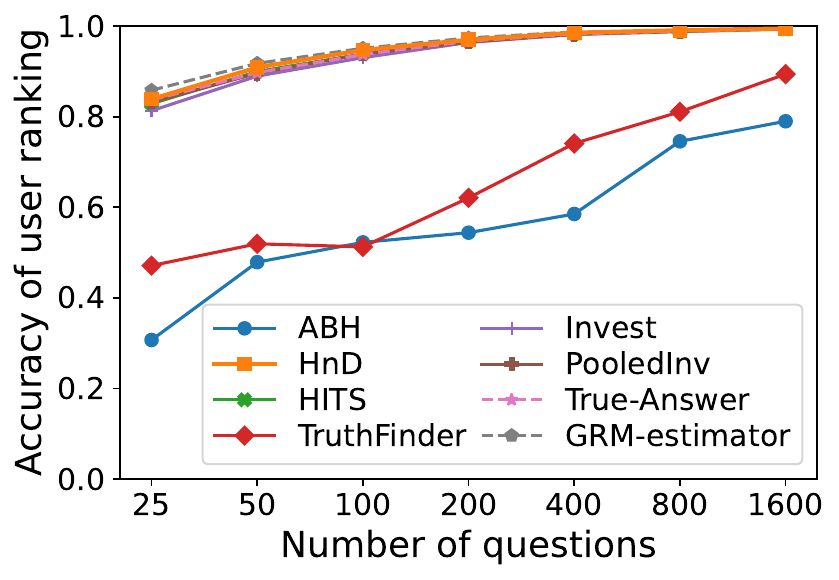}
        \caption{Varying $n$ (GRM)
		}
        \label{fig:acc_grm_q}
    \end{subfigure}
	~
	\begin{subfigure}[h]{0.23\textwidth}
        \centering
        \includegraphics[width=\textwidth]{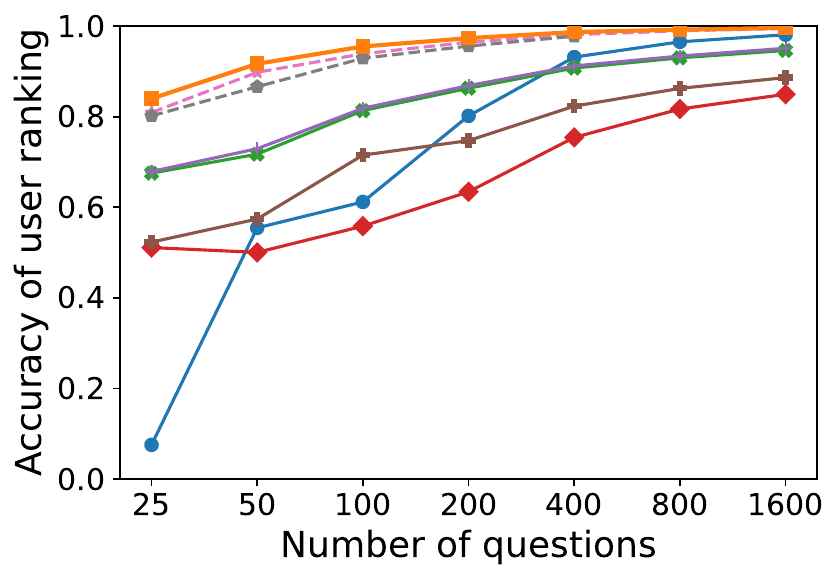}
        \caption{Varying $n$ (Bock)}
        \label{fig:acc_bock_q}
    \end{subfigure}
    ~
     \begin{subfigure}[h]{0.23\textwidth}
        \centering
        \includegraphics[width=\textwidth]{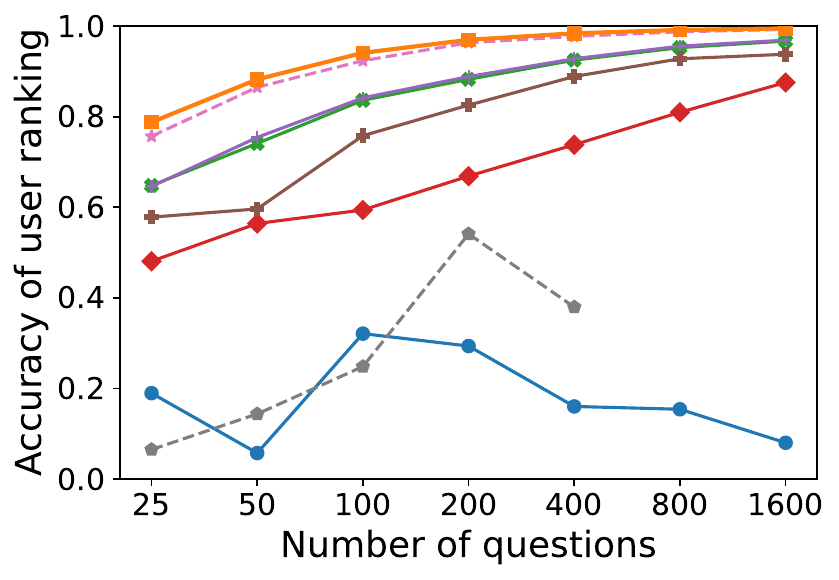}
        \caption{Varying $n$ (Samejima)}
        \label{fig:acc_samejima_q}
    \end{subfigure}
    ~
    \begin{subfigure}[h]{0.23\textwidth}
        \centering
        \includegraphics[width=\textwidth]{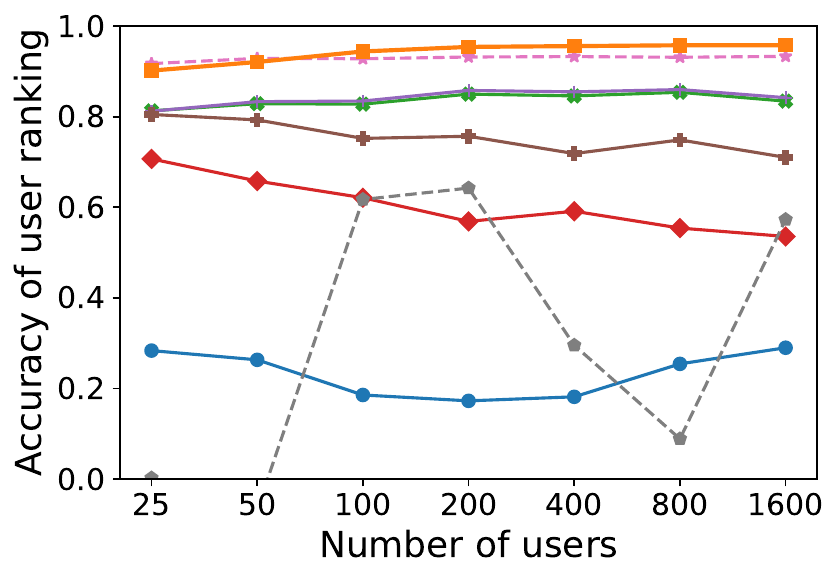}
        \caption{Varying $m$ (Samejima)}
        \label{fig:acc_samejima_u}
    \end{subfigure}
	
    \begin{subfigure}[h]{0.23\textwidth}
        \centering
        \includegraphics[width=\textwidth]{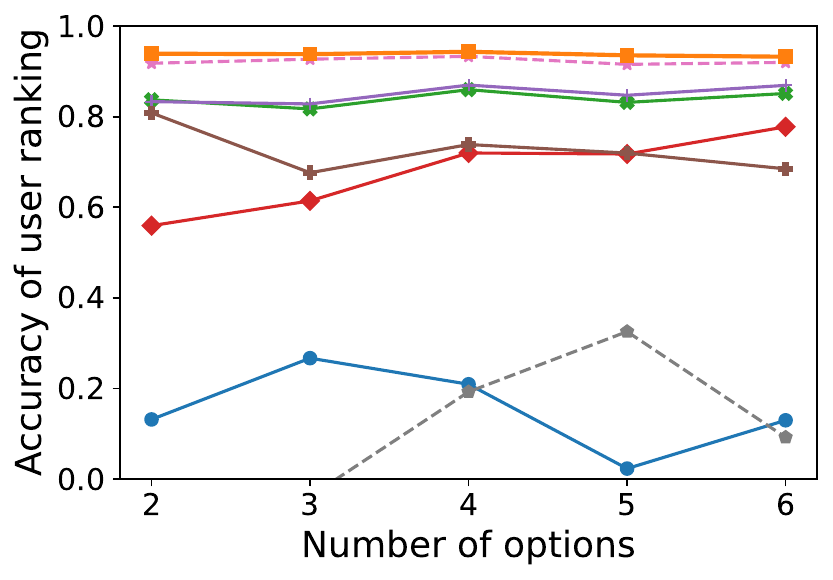}
        \caption{Varying $k$ (Samejima)}
        \label{fig:acc_samejima_o}
    \end{subfigure}
	~
    \begin{subfigure}[h]{0.23\textwidth}
		\centering
		\includegraphics[width=\textwidth]{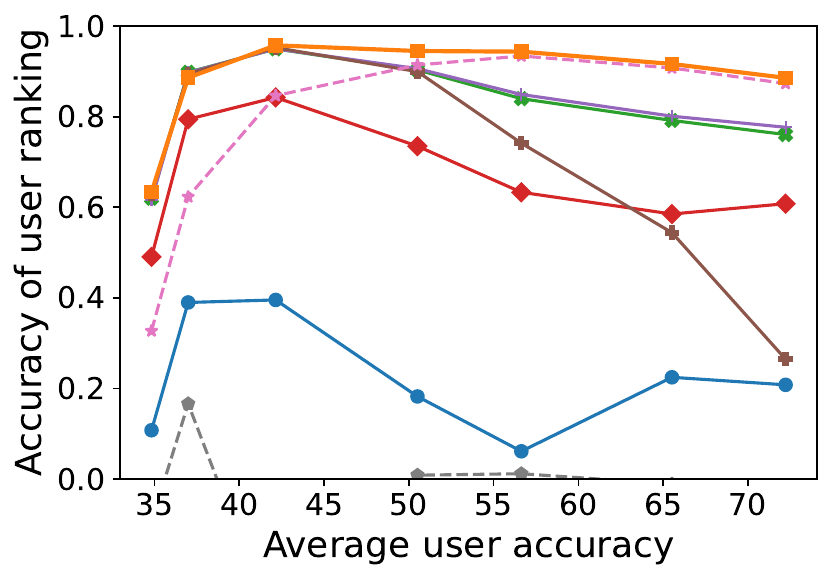}
        \caption{Varying $b_{ih}$ (Samejima)}
        \label{fig:acc_samejima_d}
    \end{subfigure}
    ~
    \begin{subfigure}[h]{0.23\textwidth}
        \centering
        \includegraphics[width=\textwidth]{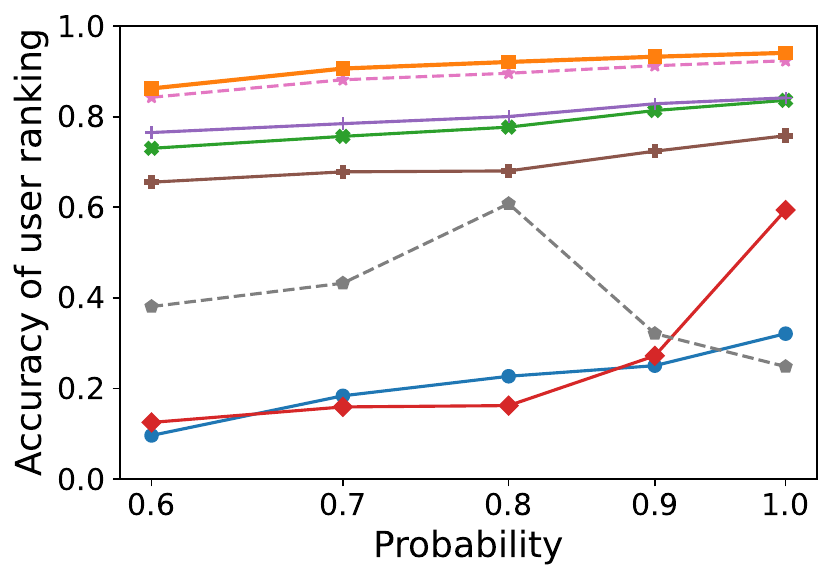}
        \caption{Varying $p$  (Samejima)}
        \label{fig:acc_samejima_p}
    \end{subfigure}
    ~
    \begin{subfigure}[h]{0.23\textwidth}
		\centering
		\includegraphics[width=\textwidth]{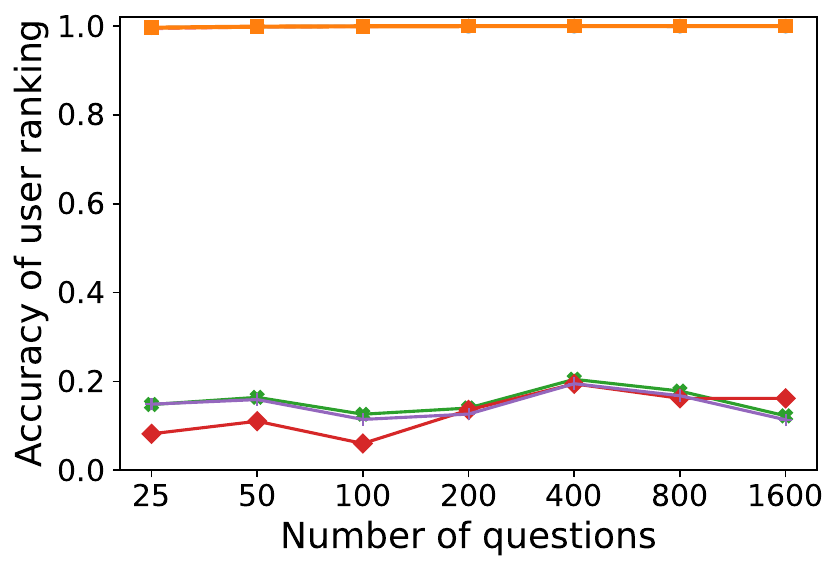}
        \caption{Varying $n$ (C1P)}
        \label{fig:acc_c1p}
    \end{subfigure}
	\caption{\Cref{sec:acc}: Results of accuracy experiments (the legend is in the first figure).	}
\label{fig:acc}
\end{figure*}

\introparagraph{Accuracy}
To determine the accuracy of a method, 
we calculate \emph{Spearman's rank correlation coefficient} 
\cite{spearman1904proof} 
between the returned user ranking and the ground truth ranking by their actual abilities.
Spearman's correlation is defined as the \emph{Pearson correlation 
between the rankings} of two scoring functions and ranges between $-1$ and~$1$.
It is similar to Kendall's correlation, yet strictly preferred if there are ties in the data~\cite{PuthSpearmanKenall:2015}.
There can be negative accuracy at times (not shown in \Cref{fig:acc}), which means the returned ranking is negatively correlated or random whose coefficient is near 0.

\introparagraph{Setup}
We conduct 
experiments to determine the accuracy as function of the
1) \emph{number of items $n$},
2) \emph{number of users $m$},
3) \emph{number of options $k$},
4) \emph{option difficulties $b_{ih}$},
and
5) \emph{probability of questions to be answered $p$}.
In the first experiment, we use data generated according to the three polytomous models (GRM, Bock, Samejima) from~\cref{sec:irt} to also show the robust performance of \HnD for all three models.
In other experiments, we only use data generated according to the Samejima model since it is the most general one to avoid redundancy.
To verify the ability of algorithms to recover a C1P ranking, 
we also generate data that 6) \emph{follows the consistent response property} 
(which as discussed in \cref{sec:c1p} is the case for IRT models when the discrimination $a_{ih} \rightarrow \infty$).
Every question has the same number of options, and every user answers every question.
By default, we set the users $m=100$, items $n=100$, and options $k=3$.

\introparagraph{1. Varying number of questions $n$
(\Cref{fig:acc_bock_q,fig:acc_samejima_q,fig:acc_grm_q}\footnote{The GRM estimator does not work when the question number is large.})} 
\HnD has better than or equal accuracy as the other methods even including the two cheating competitors over data generated by all three models.
Notice that the GRM-estimator 
works poorly for 
Samejima because it does not take random guessing into account.

\introparagraph{2. Varying number of users $m$
(\Cref{fig:acc_samejima_u})}
\HnD works here also better than or equal to other approaches
(except for the data point with low $m$ where the cheating competitors win).

\introparagraph{3. Varying number of options $k$
(\Cref{fig:acc_samejima_o})} 
\HnD stays top and accurate, even slightly outperforming True-answer.

\introparagraph{4. Varying question difficulties $b_{ih}$
(\Cref{fig:acc_samejima_d})} 
Here, we change the difficulty range from the default $[-0.5,0.5]$ to 7 different ranges, $[-1,0]$, $[-0.75,0.25]$, $[-0.5,0.5]$, $[-0.25,0.75]$, $[0,1]$, $[0.25,1.25]$, $[0.5,1.5]$,
while user abilities $\theta_j$ remain at $[0,1]$.
Thus even the least able user has a high probability to answer a difficult question
in the easiest setting,
while even the best user can be incorrect for some easy questions in the hardest setting.
The x-axis here is the \emph{average accuracy} on the questions across all the users.
In all scenarios, we see \HnD outperforms other competitors.

\introparagraph{5. Varying probability $p$ of answering a question (\Cref{fig:acc_samejima_p})}
To show that \HnD works for more general scenarios where users 
\emph{answer different number of questions}, we vary the probability of questions to be answered.
For each pair of question and user, there is the probability of $p$ for the user to answer the question.
We see that \HnD performs well even when the dataset is not complete.

\introparagraph{6. C1P (\Cref{fig:acc_c1p})}\footnote{In the experiments, PooledInv returned all negative coefficients but we consider its ranking to be the reverse one.}
In addition to the three multinomial IRT models, we also generate response matrices that are consistent and can be reconstructed to be a P-matrix.
These responses correspond to a random GRM instance with very strong discrimination $a$.
We use these matrices to verify the effectiveness of \HnD in reconstructing a C1P permutation.
To avoid ties in the rankings and provide a unique C1P ordering, we set both the user ability $\theta$ and the difficulty parameter $b$ to be within $[0,1]$, randomly chosen.
We see that \HnD and \ABH are indeed the only two methods that can reconstruct the C1P permutation if there exists one.

\introparagraph{Summary}
\HND is a \emph{robust} method that outperforms the other approaches in most setups, especially those with high discrimination. 
We see this as vindication for designing an approach based on the principle that \emph{consistent answers need to be solved correctly}.
\HnD is also competitive even against the two cheating approaches which have the best item answer given  
(i.e.\ they have access to an oracle that can solve the entire problem of truth discovery).
Moreover, we verified that \HnD and \ABH are indeed the only ones that can reconstruct a C1P permutation if it exists.

\subsection{Scalability experiments}
\label{sec:sca}

\begin{figure}[t]
\captionsetup[subfigure]{justification=centering}
    \centering
    \begin{subfigure}[h]{0.23\textwidth}
        \centering
        \includegraphics[scale=0.29]{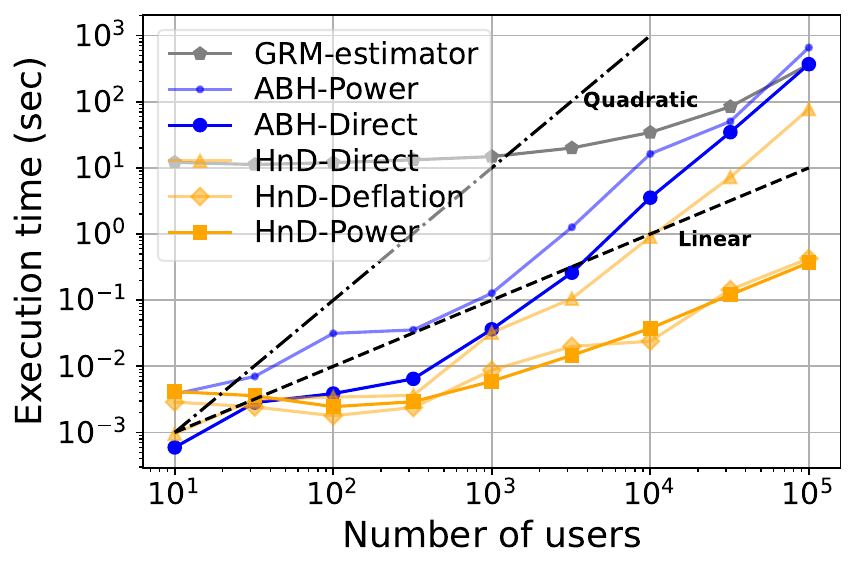}
        \caption{Scalability with users ($m$).}
        \label{fig:sca_u}
    \end{subfigure}
    \begin{subfigure}[h]{0.23\textwidth}
        \centering
        \includegraphics[scale=0.29]{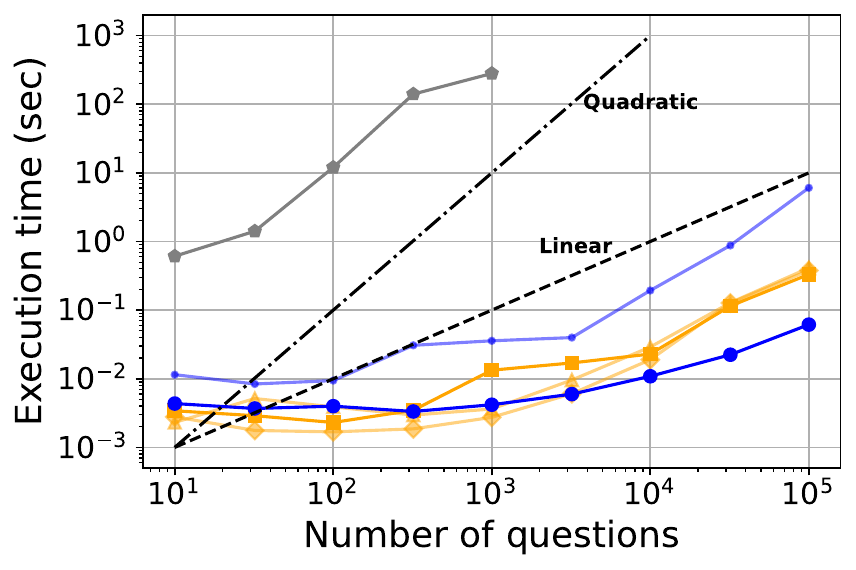}
        \caption{Scalability with items ($n$).}
        \label{fig:sca_q}
    \end{subfigure}
\caption{\Cref{sec:sca}: Scalability experiments with $n=100$ items and increasing numbers of users $m$ in (a),
or $m=100$ users and increasing numbers of items $n$ in (b).
The experiments confirm that our method (HnD) scales linearly in the number of items and users, whereas ABH (even trying various alternative methods) has an unavoidable quadratic scalability in the number of users.
}
\label{fig:sca}
\end{figure}

\Cref{fig:sca_u,fig:sca_q} show the scalability of our various implementations of \ABH and \HnD, as well as the GRM-estimator 
w.r.t.\

number of users ($m$) and
questions ($n$).
Each shown data point is the median over 5 runs, 
and we set a timeout of 1,000 seconds.
\introparagraph{\HnD vs.\ \ABH}
\Cref{fig:sca_u} shows that \ABH-direct and \HnD-direct scale with $O(m^2k)$ in the number of users as predicted in \cref{sec:complexity}.
The theoretic time complexity of \ABH-power is $O(m^2t)$ when $m$ is much larger than $n$, and thus it also takes quadratic time.
In contrast, \HnD-power can scale linearly and is about 20\% faster than \HnD-deflation on average for $m>1000$ users
as it needs only one round of the power method.
\Cref{fig:sca_q} shows that although \ABH-direct is slightly faster for fixed few users, 
all implementations are efficient even for a large number of questions.

\introparagraph{GRM-estimator}
As a representative of max likelihood parameter estimation, 
the GRM-estimator is expected to perform best on GRM data.
However, \Cref{fig:sca} shows that such parameter estimation
is by orders of magnitude
slower than \HnD .

\introparagraph{Summary}
\HnD scales asymptotically and practically better than the other existing C1P reconstruction algorithm \ABH
in the number of users.
Moreover, our intuitive \cref{alg:hnd} is slightly faster than an adaptation of the deflation method.

\subsection{Stability experiments for \ABH and \HnD}
\label{sec:stability}

\begin{figure*}[h!]
	    \captionsetup[subfigure]{justification=centering}
	    \centering
	    \begin{subfigure}[h]{0.25\textwidth}
	        \centering
	        \includegraphics[scale=0.3]{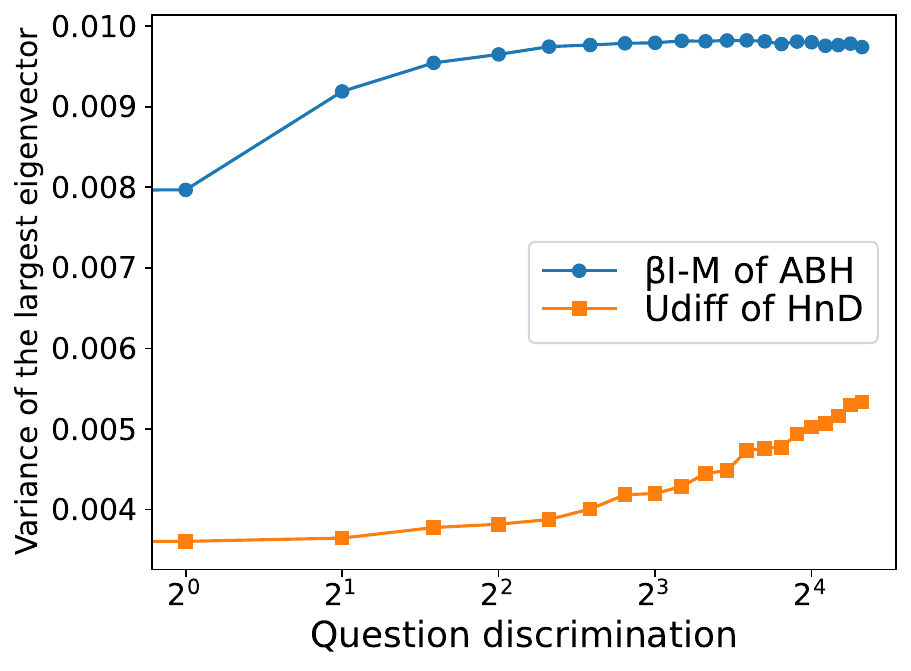}
	        \caption{Variance of eigenvector used by \HnD or \ABH
		}
	        \label{fig:comparison_variance}
	    \end{subfigure}
		~
	    \begin{subfigure}[h]{0.25\textwidth}
	        \centering
	        \includegraphics[scale=0.3]{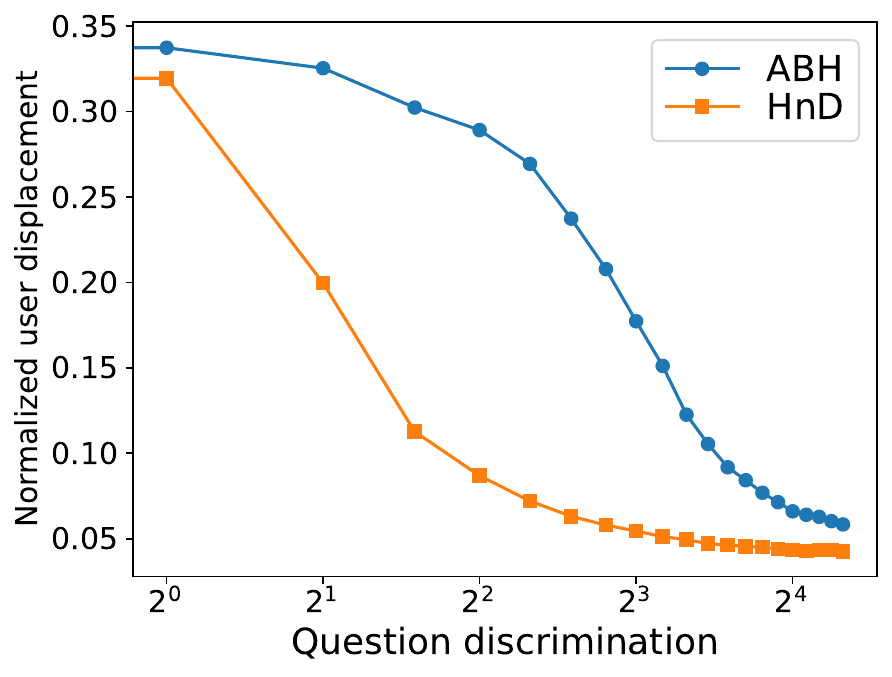}
	        \caption{Normalized user displacement
			}
	        \label{fig:comparison_difference}
	    \end{subfigure}
		~
	    \begin{subfigure}[h]{0.25\textwidth}
	        \centering
	        \includegraphics[scale=0.3]{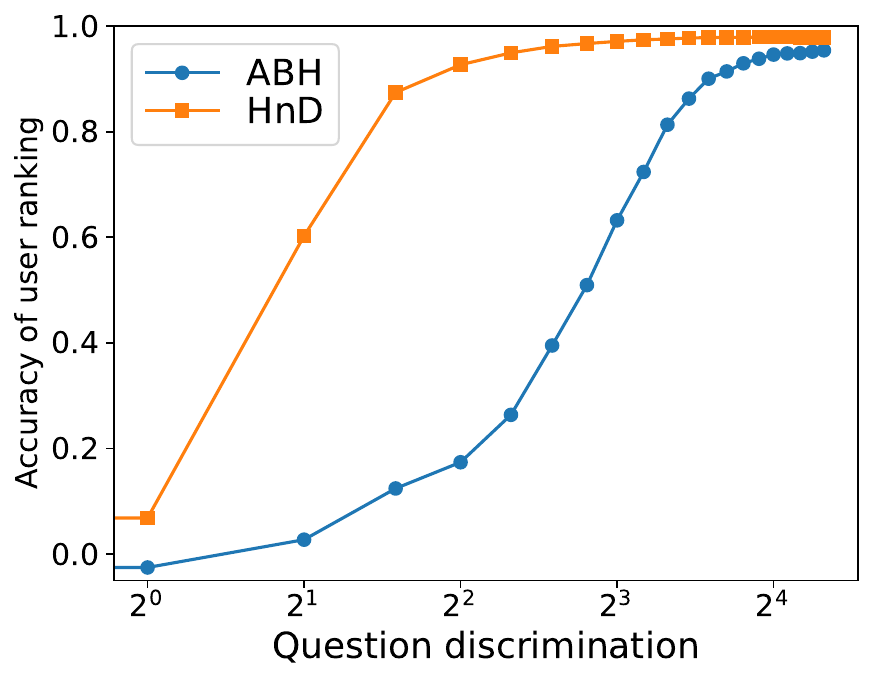}
	        \caption{Accuracy of user ranking}
	        \label{fig:comparison_accuracy}
	    \end{subfigure}
	    \caption{\Cref{sec:stability}: Stability experiments: 
	(a) The variance of the eigenvector 
	used by \HnD is smaller than that used by \ABH, 
	which makes it more robust to perturbations from the ideal C1P case. 
	This leads to \HnD
	having lower difference in the user ranking (b) and higher accuracy (c).	}	
\end{figure*}

We next experimentally verify our prediction from
\cref{sec:Hnd_vs_ABH_theory}
that \HnD generalizes better from the ideal case than \ABH.
In this setup, we fix $m=100$ users, $n=100$ items, $k=3$ options,
with user abilities and item difficulties equally spaced between $[0,1]$ and $[-0.5,0.5]$ respectively.
For one item, all the option difficulties are the same.
All items have identical discrimination $a$ and all the options in one question have equally spaced $a$ (as in the GRM model).

We then vary the question discrimination scores and compare the 
($i$) variance of the respective eigenvectors used for ranking;
($ii$) the normalized average difference in rank between each user's ranking;\footnote{
Here difference means the average difference of each user's rankings from different runs, scaled down to $[0,1]$ by the user number.}
and ($iii$) the average accuracy of the predicted rankings for \HnD and \ABH across repeatedly sampled response matrices.

\cref{fig:comparison_variance} shows our observation from \cref{sec:Hnd_vs_ABH_theory} that the variance of the largest eigenvector of $\updatediff$ of \HnD is much smaller than $\beta \vec I_{m-1} - \vec M$, which is expected to lead to the better stability and accuracy of the \HnD rankings.
\cref{fig:comparison_difference} confirms that the ranking of a user is more stable for \HnD.
\cref{fig:comparison_accuracy} shows the resulting increase in the accuracy of \HnD over \ABH.
This confirms our original goal to develop a spectral method that can achieve the same C1P ranking as \ABH,
yet generalizing better in the non-ideal case.

\subsection{Accuracy experiments on real-world data}
\label{sec:multichoice}

As mentioned in \cref{sec:benchmark}, we do not know \emph{any existing benchmark with a known true ranking} of users by their abilities.
In order to still verify the performance of \HnD on real-world datasets we use the ranking of the \emph{``True-answer"} baseline as the ground truth.
Notice that although this baseline performs well in our synthetic experiments, it is far from the perfect gold standard (sometimes even outperformed by \HnD) so the experimental result in this subsection should be seen only  as a supplementary evidence.
The six used real world MCQ datasets are from \cite{li2017hyper}.

\begin{figure}
	\centering
        \includegraphics[width=0.35\textwidth]{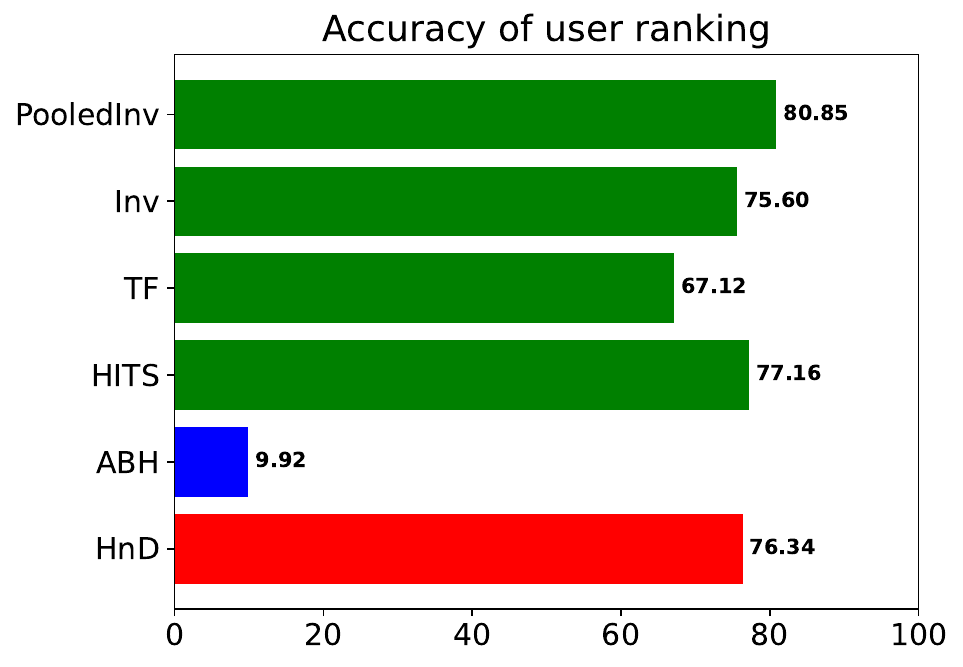}
        \caption{\cref{sec:multichoice}: Correlation of user ranking on real-world datasets with the ``True-answer'' baseline 
		that serves us as approximate gold standard user ranking.
		Notice that the true ranking of users by ability is not known.}
        \label{fig:multichoice}
\end{figure}

\cref{fig:multichoice} shows the average experimental result of six datasets where PooledInvestment and HITS perform slightly better than \HnD.
However, we need to emphasize several points: 
(1) All datasets are very small in terms of question numbers (from 20 to 36) but have double user numbers on average, which indicates their limited discrimination.
(2) There is no consistent winner on all six datasets 
\iflabelexists{sec:nomenclature}
{(see detailed result in \cref{app:more_experiments}),}
{(see detailed result in our online appendix~\cite{HITSnDIFFS-arxiv:2023}),}
an observation also made by \cite{Zheng:2017:TIC:3055540.3055547} for the related truth discovery problem.
(3) All other models except ABH with poor performance tend to have more similar accuracy than \HnD while \HnD tops them by far on 2 of the 6 datasets, which shows the novelty of \HnD and its usefulness on data of different distributions.

\section{Additional Related Work}
\label{sec:stateofart}
In this section, we discuss additional approaches for the truth discovery problem.
This is in addition to existing C1P reconstruction algorithms discussed in  
\cref{sec:c1p} and HITS-based truth discovery approaches discussed in \cref{sec:hits}.

\introparagraph{Spectral approaches}
Dalvi et al.\ \cite{dalvi} proposed two methods that output the user abilities and item labels with the help of eigenvector computation.
Ghosh et al.\  \cite{ghosh} proposed a method that only outputs the item labels, which involves calculating the first eigenvector of a symmetric matrix.
Both approaches work only for binary problems and are not obvious to generalize for $k>2$ options.

\introparagraph{Other truth discovery approaches}
\cite{li2017hyper} proposes the concept of experts and utilizes the observation that experts are more likely to reach consensus on a set of single questions (called hyper-questions in the paper) to conduct majority vote on hyper-questions instead of single questions but cannot quantify user abilities nor rank them.
\cite{lyu2019truth} relies on embeddings that cannot be easily converted into a ranking on users.
\cite{li2014confidence} uses confidence and focus on long-tail data. 
\cite{kajino2012convex, dawid} are optimization-based methods that only consider homogeneous questions (recall \cref{sec:formulation}).

\introparagraph{Other truth discovery problems}
Many approaches have been proposed for truth discovery. 
Most have different setups and are not applicable to our problem.
\cite{yang2020game, sunahase2017pairwise} change the setup of the problem by assigning different tasks to two groups of workers, where the first group answers the questions and the second group evaluates the answers.
\cite{dong2012less, rekatsinas2014characterizing, ZhengCMM15, ZhengWLCF15} work on the problem of how to assign questions to only a subset of the sources.
\cite{dong2010global, dong2009integrating} pay attention to the sources of information, 
yet focus on the copying relationships between sources. In our setup, no information is copied between users.
\cite{slimfast} discusses the problem of how much training data is needed to gain high-quality models.
\cite{ouyang2016aggregating} studies truth discovery in quantitative applications, such as percentage annotation and object counting.

\introparagraph{Crowdsourcing}
The ability discovery problem is closely connected to the truth discovery problem 
which occur in a wide range of data management problems 
related to crowdsourcing~\cite{Zheng:2017:TIC:3055540.3055547, LiWZF17}.
Various crowdsourcing systems have been proposed \cite{ZhengLC16, FanLOTF15}, 
and the crowdsourcing approach has been refined for various tasks, such as query answering \cite{franklin2011crowddb}, 
entity resolution\cite{ChaiLLDF16}, annotating Twitter data \cite{finin2010annotating}, top-k algorithms\cite{ZhangLF16} 
and various other labeling tasks \cite{tarasov2010using, welinder2010online, HuZBLFC16}.

\introparagraph{Expert finding}
The expert finding problem \cite{LinHWL17, YuanZTHC20} also aims to assess the trustworthiness of users.
The difference is that it focuses on \emph{finding experts with expertise (skills) specific to a given question} while our ability discovery problem aims to assess an overall user ability.

\section{Conclusions}
\label{sec:conclusion}
We proposed \HND, a novel variant of HITS,
with surprising theoretical and practical properties for ability discovery.
On the theoretical side, we showed that 
1) \cons of the response matrix models consistent solutions for the problem;
2) our method reconstructs the correct user rankings in the consistent case;
3) does so in linear time and
4) can handle more general cases (in contrast to other linear discrete algorithms).
On the practical side, we showed that 
\HnD handles the problem of ability discovery with robust accuracy and greater scalability in terms of the number of users than the only existing C1P reconstruction algorithm that works for general cases.

\section*{Acknowledgment}
This work was supported in part by the National Science Foundation (NSF) under award numbers IIS-1762268 and IIS-1956096.

\newpage
\bibliographystyle{IEEEtranS} 
\bibliography{truthdiscovery}

\clearpage
\appendices
\section{Nomenclature}\label{sec:nomenclature}

\begin{table}[h]
    \centering%
    \small
    \begin{tabularx}{\linewidth}{  @{\hspace{0pt}} >{$}l<{$}  @{\hspace{2mm}}  X @{}}
    \hline
    m & number of users\\
    n & number of items\\
    k & number of choices (or options) for each item\\
    \user_j & user $j$ \\
    \theta_j & ability of user $j$\\
    t_i & item $i$\\
    c_{ih}	& 	choice (or option) $h$ for item $i$ \\
    a_{ih} 		& discrimination of option $h$ for item $i$ in discrimination-difficulty parameterization \\	
    b_{ih} 		& difficulty of option $h$ for item $i$	in discrimination-difficulty parameterization\\
    \alpha_{ih} & discrimination (or slope) of option $h$ for item $i$ in slope-intercept parameterization \\
    \beta_{ih} 	& difficulty (or intercept) of option $h$ for item $i$ in slope-intercept parameterization\\	
    c_i 		& random guessing parameter in some IRT models for item $i$\\
    \vec{A}_{j:} & $j^{\mathrm{th}}$ row of some matrix $\vec{A}$ \\
    \vec{A}_{:i} & $i^{\mathrm{th}}$ column of some matrix $\vec{A}$ \\
    {A}_{ji} 	& $(j,i)$ entry of matrix $\vec{A}$ \\
    \lambda_i 	& the $i$-th largest eigenvalue\\
    \vec v_i & the eigenvector corresponding to the $i$-th largest eigenvalue (for simplicity, also referred to as ``$i$-th largest eigenvector'')\\
    \response 		& $(m \times kn)$ binary response matrix  \\
    \vec D & the  diagonal matrix of a square matrix. 
		Of size ($m \times m$) when it represents the diagonal matrix 
		of $\response\response^\transpose$ used by \ABH~\cite{ABH} \\
    \vec L &  the Laplacian matrix of a square matrix. 
		Of size ($m \times m$) when it represents the Laplacian matrix 
		of $\response\response^\transpose$ used by \ABH~\cite{ABH}: 
		$\laplacian = \vec{D} - \response \response^\transpose$ \\
     \userscore 	& $(m)$ user score vector where $s_j$ denotes score of user $j$\\
     \qweight 	& $(kn)$ option weight vector where $w_{i}$ denotes the weight of 
		option $((i\!-\!1) \textrm{ mod } k) +1$ 
	for item $((i\!-\!1) \textrm{ div } k) + 1$ \\
    \onesvec & the ones vector $\vec{1}_r$ for some constant $r \in \N^+$ \\
	\responsecol 	& $(m \times kn)$ column-normalized response matrix  \\
    \responserow 	& $(m \times kn)$ row-normalized response matrix \\
    \update 		& $(m \times m)$ update matrix  $\update = \responserow(\responsecol)^\transpose$
		used by \avgHITS  \\
    \updatediff & $((m\!-\!1) \times (m\!-\!1))$ difference update matrix $\updatediff = \vec{S} \update \vec{T}$
		used  by \HnD to update user score differences $\userdiff$\\
	\userdiff 	& $(m\!-\!1)$ user difference vector with 
	$\userdiff_j = \userscore_{j+1} - \userscore_{j}, j \in [m\!-\!1]$ \\
	 \vec S & $((m\!-\!1) \times m)$ matrix to compute the differences in a vector\\
    \vec T & $(m \times (m\!-\!1))$ matrix to reconstruct a vector based on the differences\\
	\vec M & $((m\!-\!1) \times (m\!-\!1))$ matrix $\vec M = \vec S\vec L\vec T$ \\
    \vec I & identity matrix (quadratic matrix with 1's on the diagonal)\\
	t		& number of iterations \\
    \vec{e}_i & the vector with the $i^{\mathrm{th}}$ element $1$ and all other elements $0$ \\
    {}
    {}[m]		& set $\{1, 2, \ldots, m \}$ \\
    \hline
    \end{tabularx}
\end{table}

\section{Proofs}
\label{sec:ideal}
In this section, we prove our main \cref{theorem: main}.
Before the proof, we first review some background.\footnote{The references in the appendix refer to the full reference, which includes 5 more papers.}

\introparagraph{The method of Atkins et al. (\ABH)}
Given a pre-P matrix $\response$ with a unique consecutive ones ordering,
the spectral method of Atkins et al.\ \cite{ABH} (which we refer to as \ABH)
finds this ordering of users that reconstruct the \cons property for $\response$
by using
the product of the response matrix $\response$ and its transpose to form $\response \response^\transpose$.
The method has two main parts: showing that if $\response$ is a P-matrix,
then (1) $\response  \response^\transpose$ has the special property of being an R-matrix (which we will define below), and
(2) the \emph{2nd smallest} eigenvector
of  the related Laplacian matrix $\laplacian = \vec{D} - \response \response^\transpose$ is monotonic
(i.e.\ its elements are in either increasing or decreasing order).

Here $\vec{D}$ is a diagonal matrix filled with the row sums of $\response \response^\transpose$.
Given these two facts, they prove that if $\response$ is a pre-P matrix, then the row permutations induced by sorting the values of 2nd smallest eigenvector of $\laplacian$ can reconstruct the \cons property for $\response$.

\begin{definition}[R-matrix~\cite{ABH}]
A matrix $\vec{A}$
is called an \emph{R-matrix} if it is symmetric and
\begin{align*}
A_{ji} &\geq A_{jh} \hspace{5mm}\textrm{for}\hspace{1mm}  j < i < h	\\
A_{ji} &\leq A_{jh} \hspace{5mm}\textrm{for}\hspace{1mm}  i < h < j 
\end{align*}
\end{definition}

Intuitively, the definition describes a matrix where the values fall off as we move away from the diagonal along any row.
Since for $ \response \response^\transpose$, each matrix entry represents the dot product of the responses of two users. 
The entries represent the number of common responses for a pair of users. 
When the response matrix is sorted by user ability, these values are highest for users with themselves and fall off as we move away along the correct linear ordering of the abilities represented in the correctly sorted response matrix.

\introparagraph{Our method}
Given a pre-P matrix $\response$ with a unique consecutive ones ordering, 
we use our update matrix 
$\update = \responserow(\responsecol)^\transpose$ 
to find this orderings of users that reconstruct the \cons property for the matrix $\response$.
To do this, we will prove the following two statements: 
(1) \cref{theorem: R matrix}
shows that if $\response$ is a P-matrix, then $\update$ is an R-matrix. 
(2) \cref{theorem: monotonic} shows that the \emph{2nd largest} eigenvector of $\update$ is monotonic. 
These two statements imply that if $\response$ is a pre-P matrix, then the eigenvector 
corresponding to the 2nd largest eigenvalue of $\update$ can be used to find the unique ordering that reconstruct the \cons property for the matrix $\response$.

Without loss of generality, we assume that every item option was chosen by at least one student.
If an option was never chosen, this option has no information and can be removed.
Thus in the following proof, we assume that every column of $\response$ contains at least one $1$.

We first recall a well-known lemma that we will use in the proofs.
\begin{lemma}[Constant row sums]
\label{theorem: all ones}
If all rows of a non-negative square matrix $\textbf{A}$ sum to a scalar $b$, then the largest eigenvalue of $\textbf{A}$ is $b$ with the corresponding eigenvector in the direction of $\onesvec = \vec{1}_n$.
\end{lemma}

\begin{proof}[Proof of \cref{theorem: all ones}]
	It is easy to see that $\vec e$ is an eigenvector with eigenvalue $b$: 
	since every row in $\vec A$ sums to $b$, every entry in $\vec A \vec e$ is $b$.
	That $b$ is the largest eigenvalue follows from 
	the fact that the spectral radius $\rho(\vec A) \leq ||\vec A||$ for every matrix norm~\cite{Meyer00},
	including the induced $\infty$-norm $|| \vec A ||_{\infty}$	
	which is the maximum absolute row sum.
\end{proof}

Note that if there are multiple connected components in the user-option bipartite graph, there is no way to get a total ordering on the users or items, since we cannot compare between the different connected components; we can only get an ordering for each connected component
\emph{separately}. Thus, in the sequence, we assume a single connected component, and thus, multiplicity $1$ of the largest eigenvalue of the update matrix $\update$.

\begin{lemma}[$\update$ is row-stochastic]\label{lemma:rowsum}
Each row of $\update$ has sum $1$.
\end{lemma}

\begin{proof}[Proof of \cref{lemma:rowsum}]
Define $r_j$ as the sum of elements of the $j^\textrm{th}$ row of 
$\update = \responserow(\responsecol)^\transpose$,
i.e.\ $r_j = \sum_{i=1}^{m}{U_{ji}}$.
We want to show that $r_j = 1, \forall j \in [m]$.

Since $\update = \responserow(\responsecol)^\transpose$,
the $j^\textrm{th}$ row of $\update$, $\vec{U}_{j:}$, 
is created from the corresponding row $\responserow_{j:}$ of $\responserow$ and all the rows $\responsecol_{i:}$ of $\responsecol$:
\begin{align*}
  r_j &= \sum_{i=1}^{m}{\textbf{C}^\textrm{row}_{j:}} \cdot ((\textbf{C}^\textrm{col})^\transpose)_{:i} 
  	= \sum_{i=1}^{m}{\textbf{C}^\textrm{row}_{j:}} \cdot \textbf{C}^\textrm{col}_{i:} 
  	= {\textbf{C}^\textrm{row}_{j:}} \cdot \sum_{h=1}^{m} \textbf{C}^\textrm{col}_{i:}
\end{align*}
But recall that by construction, $\sum_{i=1}^{m} \textbf{C}^\textrm{col}_{i:}$ is the ones vector $\onesvec$
(every option is chosen by at least one student or otherwise the whole column is removed). 
So,
\begin{align*}
  r_j = {\textbf{C}^\textrm{row}_{j:}} \cdot \sum_{i=1}^{m} \textbf{C}^\textrm{col}_{i:} = {\textbf{C}^\textrm{row}_{j:}} \cdot \onesvec = 1
\end{align*}

The last step follows from the definition of the row-normalized binary response matrix $\responserow$.
\end{proof}

\begin{lemma}[1st eigenvector of \avgHITS]\label{lemma: fixed point of U}
If the largest eigenvalue of $\update$ has multiplicity 1, the fixed point of the \avgHITS update rule is in the direction of $\onesvec = \vec{1}_m$
\end{lemma}
\begin{proof}[Proof of \cref{lemma: fixed point of U}]
The \avgHITS update rules above imply for any iteration $t > 0$,
\begin{align*}
\textbf{s}^{(t)} &= \update \userscore^{(t-1)} \\
\implies \textbf{s}^{(t)} &= \update^t \userscore^{(0)}
\end{align*}

Let $\lambda_1$ be the largest eigenvalue of $\update$ with the corresponding eigenvector $\vec v_1$. 
By the power rule, if $\lambda_1$ has multiplicity $1$, 
$\textbf{s}^{(t)}$ converges in the direction of $\vec v_1$ as $t$ goes to infinity 
(barring the very unlikely random initialization of $\userscore^{(0)}$ that is orthogonal to $\vec v_1$).

Note that $\update$ is square and non-negative since $\responserow$ and $\responsecol$ are both non-negative. 
By \cref{lemma:rowsum}, each row of $\update$ has sum $1$, and we can then use \cref{theorem: all ones}.
This gives us the result, since the fixed point of \avgHITS is 
in the direction of the eigenvector corresponding to the largest eigenvalue $\lambda_1$ of $\update$ 
(provided $\lambda_1$ has multiplicity 1) 
and \cref{theorem: all ones} implies $\update$ has largest eigenvalue $1$ with the corresponding eigenvector $\onesvec$.
\end{proof}

We can now prove our main result in \cref{theorem: unique} using the following three lemmas.
In those lemmas, we make the simplifying assumption that each row of the matrix $\response$ has the same row sum (i.e.\ the same number of $1$s in each row), which means each user chooses the same number of options.
This is WLOG, because given any arbitrary matrix $\response$, we can add additional columns with exactly one $1$ and $m-1$ zeros each to the matrix $\response$ until each row has the same number of $1$s. 
Note that adding such columns does not affect the \cons property, since each additional column has only a single $1$. 
In practice, the ordering of users might be affected by padding the matrix with columns in this way, so we only make this assumption of equal row sums for proving results related to the \cons property of the ideal case.

\begin{lemma}
\label{theorem: U symmetric}
If the response matrix $\response$ is a P-matrix and each row has the same row sum, the update matrix $\update$ is symmetric.
\end{lemma}

\begin{proof}[Proof of \cref{theorem: U symmetric}]
First we assume each user chooses $n^{ans}$ options.
Consider any $i,j \in [m], i \neq j$. We need to show that $U_{ij} = U_{ji}$, where
\begin{align*}
U_{ji} &= \sum_{h = 1}^{\textrm{width}(\response)} {C^\textrm{row}_{jh} C^\textrm{col}_{ih}} \\
U_{ij} &= \sum_{h = 1}^{\textrm{width}(\response)} {C^\textrm{row}_{ih} C^\textrm{col}_{jh}} 
\end{align*}

To contribute to the sum, the respective elements of $\responserow$ and $\responsecol$ have to both be nonzero.
For all $i \neq j$, $C^\textrm{row}_{jh}$ and $C^\textrm{col}_{ih}$ being both nonzero implies $C_{jh}$ and $C_{ih}$ are 1, which also means $C^\textrm{row}_{ih}$ and $C^\textrm{col}_{jh}$ are both nonzero.

Because each user answers the same number of questions $n^{ans}$,
all nonzero entries in $\responserow$ are equal, and therefore also
$C^\textrm{row}_{jh} = C^\textrm{row}_{ih} = \frac{1}{n^{ans}}$ if both are nonzero.
Also by definition of $\responsecol$, 
$C^\textrm{col}_{jh} = C^\textrm{col}_{ih}$ if they are both nonzero 
since they appear in the same column of $\responsecol$.
Combining the two equations,
$C^\textrm{row}_{jh} C^\textrm{col}_{ih} = C^\textrm{row}_{ih} C^\textrm{col}_{jh}$
for all $i \neq j$.
Therefore, also
$U_{ij} = U_{ji}$ for all $i \neq j$, 
and thus $\update$ is symmetric.
\end{proof}

\begin{lemma}
\label{theorem: R matrix}
If the response matrix $\response$ is a P-matrix and each row has the same row sum, then the update matrix $\update$ is an R-matrix.
\end{lemma}

\begin{proof}[Proof of \cref{theorem: R matrix}]
We need to show that neither of
\begin{align}
U_{ji} > U_{jk}  \hspace{5mm}\textrm{for}\hspace{1mm}    i < k < j 		\label{eq:R matrix 1}  \\
U_{ji} < U_{jk}  \hspace{5mm}\textrm{for}\hspace{1mm}    j < i < k			\label{eq:R matrix 2}
\end{align}
is true, where
\begin{align*}
U_{ji} &= \sum_{h = 1}^{\textrm{width}(\response)} {C^\textrm{row}_{jh}C^\textrm{col}_{ih}} \\
U_{jk} &= \sum_{h = 1}^{\textrm{width}(\response)} {C^\textrm{row}_{jh}C^\textrm{col}_{kh}}
\end{align*}

Analogous to the proof for the symmetry (\cref{theorem: U symmetric}), 
we have $C^\textrm{row}_{jh} C^\textrm{col}_{ih} =  C^\textrm{row}_{jh} C^\textrm{col}_{kh}$ for all nonzero 
$C_{jh}, C_{ih}, C_{kh}$. 
In order for \Cref{eq:R matrix 1} to hold, there would have to be more non-zero 
$C^\textrm{row}_{jh} C^\textrm{col}_{ih}$ than 
$C^\textrm{row}_{jh} C^\textrm{col}_{kh}$.
This in turn would imply that among the columns of $\response$ 
which have a $1$ in the $j$th row (and therefore nonzero entries in the $j$th row of $\responserow$)
there are strictly more $1$'s in the $i$th row than the $k$th row.
But since $i < k < j$, the column which has a $1$ in the $i$th row 
but not the $k$th row leads to a violation of the consecutive ones property. 

\Cref{eq:R matrix 2} leads to a similar contradiction.
\end{proof}

\begin{lemma}
\label{theorem: monotonic}
If $\response$ is a P-matrix and each row has the same row sum, the 2nd largest eigenvector of $\update$ is monotonic.
\end{lemma}

\begin{proof}[Proof of \cref{theorem: monotonic}]
We will prove this in a way similar to \cite{ABH}.
Define the matrices $\vec{S} \in \Rset^{(m-1) \times m}$ and $\vec{T} \in \Rset^{m \times (m-1)}$ as in \cref{sec:HND_algo}.

Note that $\vec{T} \vec{S} = (\vec{I}_m - \vec{e}\vec{e}_1^\transpose)$. Note that for any vector $\vec{v}$, $\vec{S} \vec{v}$ = $(v_2 - v_1, v_3 - v_2, ..., v_r - v_{r-1})^\transpose$. Similarly, the $i^{\textrm{th}}$ row of $\vec{S} \update$ is just the difference between the $(i+1)^\textrm{th}$ and $i^\textrm{th}$ rows of $\update$. However, as shown in \cref{lemma:rowsum}, each row of $\update$ sums to $1$. This implies each row of $\vec{S} \update$ sums to $0$.

Let $\vec{x}$ be an eigenvector of $\update$ that is not in the direction of the all ones vector $\vec{e} = \vec{1}_m$, i.e. $\vec{x} \neq \alpha \vec{e}$. Then,
\begin{align}
\update \vec{x} &= \lambda \vec{x}  				\notag \\
\vec{S} \update \vec{x} &= \lambda \vec{S} \vec{x}  \notag \\
\vec{S} \update (\vec{I}_m - \vec{e}\vec{e}_1^\transpose) \vec{x} &= \lambda \vec{S} \vec{x}  \notag \\
\vec{S} \update \vec{T} \vec{S} \vec{x} &= \lambda \vec{S} \vec{x}  \notag \\
\updatediff \vec{y} &= \lambda \vec{y}, \quad \textrm{where} \quad \vec{y} = \vec{S} \vec{x} 	
\label{eq:s to s_diff}
\end{align}

The equivalence between the 2nd and 3rd lines above is 
because each row of $\vec{S} \update$ sums to $0$ as shown above, so $\vec{S} \update \onesvec = \vec{0}$. So, any eigenvalue $\lambda$ of $\update$ is also an eigenvalue of $\updatediff$, except for the eigenvalue of $\update$ corresponding to the $\onesvec$ eigenvector leading to $\vec{S} \vec{e}  = 0$. Therefore, $\updatediff$ has exactly the same eigenvalues as $\update$ except the largest eigenvalue $1$, and the eigenvectors of $\updatediff$ are the differences between the entries of the corresponding eigenvector of $\update$, which proves \cref{lemma:EigenvectorCorrespondence}.

Then we prove that $\updatediff$ is a non-negative matrix.
First, we think about $\vec{S} \update$.
As we explained before, every row of it sums to 0.
Moreover, for each the $r$th row of $\vec{S} \update$ which is the difference between the $r$th and $(r+1)$th row of $\update$, the first $r$ entries are non-positive while all the others are non-negative because $\update$ is an R-matrix.

Then we turn to the multiplication between $\vec{S} \update$ and $\vec{T}$, which generates $\updatediff$.
The entry at the $j$-th row and the $i$-th column of $\updatediff$ is the product of the $j$-th row of $\vec{S} \update$ and the $i$-th column of $\vec{T}$.
This product is actually the sum of the last $(m - i)$ entries of the $j$-th row of $\vec{S} \update$,
which also equals to the sum of the $j$-th row of $\vec{S} \update$ minus the first $i$ entries.
Recall that the first $r$ entries are non-positive while all the others are non-negative for the $r$-th row of $\vec{S} \update$, and each row of $\vec{S} \update$ sums to 0.
If $j \geq i$, the first $i$ entries of the $j$-th row of $\vec{S} \update$ 
are all non-positive so the product is non-negative. 
If $j < i$,  the last $m - i$ entries of the $j$-th row of $\vec{S} \update$ are all non-negative so the product is non-negative.

Therefore, we prove every entry in $\updatediff$ is non-negative, which means $\updatediff$ is a non-negative matrix.

We can now apply the Perron-Frobenius Theorem \cite{perron1907theorie, frobenius1912matrizen}: 
there exists a non-negative eigenvector of $\updatediff$ corresponding to the largest eigenvalue of $\updatediff$. 
Note that $\updatediff$ has exactly the same eigenvalues as $\update$, 
except the largest eigenvalue $1$, and the eigenvectors of $\updatediff$ are 
the differences between the elements of the corresponding eigenvector of $\update$.
Since the differences between the elements of the eigenvector corresponding to the 2nd largest eigenvalue of $\update$ (largest eigenvalue of $\updatediff$) are non-negative, 
that eigenvector of $\update$ is monotonic.
\end{proof}

The above thre lemmas imply that if we start with a pre-P matrix $\response$ with a consistent row sum, 
sorting the rows according to the second eigenvector ordering of the corresponding update matrix $\update$ gives a P-matrix, 
proving \cref{theorem: unique}.

From \cref{lemma:EigenvectorCorrespondence}, we know by converting the converged largest eigenvector of $\updatediff$ back into a user score, we regain the ordering of the rows according to values in the second largest eigenvector of $\update$.
This, along with \cref{theorem: unique}, proves \cref{theorem: main} that an algorithm for computing the largest eigenvector of $\updatediff$ (\HnD detailed in \cref{alg:hnd}) reconstructs the ideal consistent ordering.

\section{Item Response Theory (IRT) models}
\label{app:IRT}

\Cref{sec:IRT_binary} starts with the binary (or dichotomous) models that model the probability that a student answers a question correctly. 
\Cref{sec:multi_IRT} then discusses multinomial (or polytomous) models 
that model the probability for a student to choose a specific option of a question.

\subsection{Binary (or Dichotomous) Models}
\label{sec:IRT_binary}

The binary IRT models are variations of binary logistic regression with one latent trait $x$:
$\P(x) = \sigma(\alpha x + \beta)$.
The intuition is that the probability of a correct answer increases with a student's ability and decreases with the question's difficulty.
The \emph{ability} of a student 
is captured by a single latent factor $\theta$.

\introparagraph{1PL (1-Parameter Logistic) model}
This is the simplest logistic IRT model 
and is also called the \emph{Rasch model}~\cite{rasch1993probabilistic}.
It is called 1PL as it has only one parameter $b_i$ called difficulty for every question $i$
which shifts the logistic function with regard to student abilities.
The probability for a user with ability $\theta$ to answer a question $i$ correctly is given by the following response function:
\begin{align*}
	\P_{i}(\theta) = \sigma\big(\theta-b_i\big) = \frac{1}{1 + e^{-(\theta - b_i)}}
\end{align*}
Thus the probability of getting a question correct increases with
a larger student ability $\theta$ 
and a smaller question difficulty $b_i$.

\introparagraph{2PL (2-Parameter Logistic) model \cite{birnbaum1968}}
The 2PL model adds a second parameter $a_i$ called discrimination to every question $i$.
Intuitively, the discrimination models how well a question can separate competent from less competent students:
High discrimination implies that 
the probability of answering it correctly increases strongly with student ability $\theta$.
The new response function is:
\begin{align*}
	\P_{i}(\theta) = \sigma\big(a_i(\theta-b_i)\big) = \frac{1}{1 + e^{-a_i(\theta - b_i)}}	
\end{align*}
Notice this is exactly the logistic function 
after changing
from \emph{slope-intercept} $\sigma(\alpha_i \theta + \beta_i)$ 
to \emph{discrimination-difficulty} parameterization $\sigma\big(a_i(\theta-b_i)\big)$.
In other words, the 2PL model is \emph{isomorph to logistic regression}.
It specializes into the 1PL model by tying all discrimination parameters to the same value $a_i = 1$.

\introparagraph{GLAD}
This model was proposed in the crowdsourcing literature~\cite{NIPS2009:Whitehill}
and also uses the logistic function.
Interestingly, it is a specialization of the 2PL model with all difficulty parameters tied by $b_i = 0$.
A special property of the GLAD model is that for a student whose ability is 0, the probability to answer any question correctly is 50\%.
The new response function is:
\begin{align*}
	\P_{i}(\theta) = \sigma\big(a_i \theta\big) = \frac{1}{1 + e^{-a_i\theta}}
\end{align*}

\introparagraph{3PL (3-Parameter Logistic) model \cite{birnbaum1968}}
The 3PL model adds a third parameter $c_i$ to each question $i$ that model the probability that a student without low ability can randomly guess the correct answer.
This is motivated in a multiple-choice question (MCQ) setting where the best strategy for a student who does not know the answer is to pick one answer randomly.
A reasonable setting for $c$ is $1/k$, where $k$ is the number of choices (or options).
The 3PL model specializes to the 2PL model by tying all random guesses to 0: $c_i=0$.
The new response function is
\begin{align*}
	\P_{i}(\theta) = c_i + (1-c_i)\sigma\big(a_i(\theta-b_i)\big) = c_i + \frac{1 - c_i}{1 + e^{-a_i(\theta-b_i)}}	
\end{align*}

\subsection{Multinomial (Polytomous) Models}
\label{sec:multi_IRT}

Multinomial IRT models model the probability of a student picking a particular choice among $k$ options.
They rely on a generalisation of the logistic function to multiple inputs called the softmax activation function (or multinomial logit),
used in multinomial logistic regression.
The softmax $\sigma : \R^k \rightarrow [0 , 1]^k$ defines a probability of a choice among several outcomes 
for $k\geq 2$ as 
$\sigma(\vec x)_h = e^{x_h}/ \sum_{l=1}^k e^{x_l}$
for $h=1 \ldots k$.
The various multinomial IRT models can thus be seen as variants of multinomial logistic regression.
We next discuss three important polytomous models and their connections.

\introparagraph{Bock}
Bock's nominal category model~\cite{bock1972estimating} is exactly multinomial logistic regression in 
slope-intersection parameterization.
In this model, each option $h$ has a discrimination parameter $\alpha_h$ and an intercept $\beta_h$.
The probability for a user with ability $\theta$ to pick option $h$ for question $i$ is given by the following response function:
\begin{align}
	\P_{ih}(\theta) = \frac{e^{\alpha_{ih}\theta+\beta_{ih}}}{\sum_{l=0}^{k-1}e^{\alpha_{il}\theta+\beta_{il}}}
\end{align}
The option with the largest $\alpha$ is the correct option 
(students with large enough ability $\theta$ will always choose that option).
Also, notice that the function is actually over-parameterized: 
Dividing by $e^{\alpha_{i0}\theta+\beta_{i0}}$ gives a normalized representation with $2(k-1)$ independent parameters per question.
With the same normalization, Bock recovers the binary 2PL model for $k=2$.
(see e.g.\ ~\cite[section 2.5.3]{pml1Book}).

\introparagraph{Graded Response Model (GRM) \cite{samejima1997graded, samejima1969estimation, samejima1972general}}
The graded-response model deals with \emph{ordered} polytomous categories such as ratings
(e.g.\ strongly disagree, disagree, agree, and strongly agree, used in attitude surveys).
GRM postulates that there are $k+1$ steps per question ($0, \ldots, k$).
A student choosing an option $h \in \{0, \ldots, k-1\}$ passes the first $h+1$ steps but fails in step $h+2$.
Thus every student passes the first step but no student passes the last step.
The model uses the 2PL response function as a cumulative probability function for a student with ability $\theta$ 
to pass step $h$ as:
\begin{align*}
	&\P^*_{ih}(\theta) =  \sigma\big(a_i(\theta-b_{ih})\big) = \frac{1}{1 + e^{-a_i(\theta - b_{ih})}}	\\ 
	&-\infty = b_{i0} < b_{i1} < ... < b_{i,k-1} < b_{ik} = \infty
\end{align*}
where 
$\P^*_{ih}(\theta) = \sum_{l=h}^k \P_{il}(\theta)$. 
The option response function is then
$\P_{ih} = \P^*_{ih} - \P^*_{i,h+1}$.
For a question $i$ with $k$ options, there are $k$ free parameters in all including a discrimination parameter $a_i$ and $k-1$ difficulty
or location parameters $b_{ih}, h \in \{1, \ldots, k-1\}$.

\begin{figure}[t]
\centering
\begin{subfigure}{0.49\linewidth}
	\centering
	\includegraphics[scale=0.45]{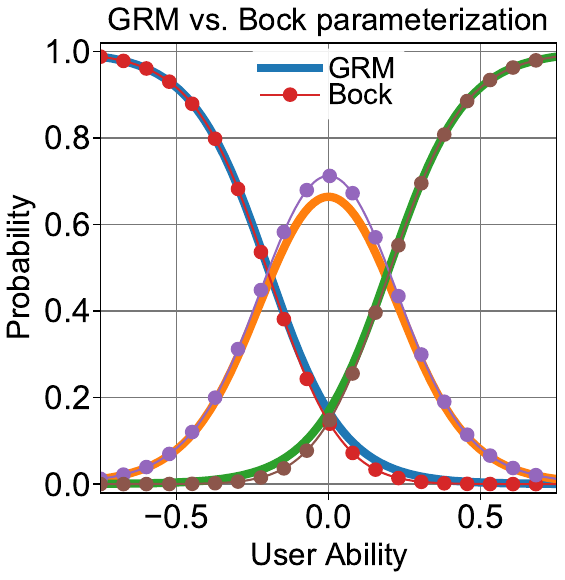}
	\caption{}
	\label{Fig_GRM_vs_Bock1}	
\end{subfigure}
\begin{subfigure}{0.49\linewidth}
	\centering
	\includegraphics[scale=0.45]{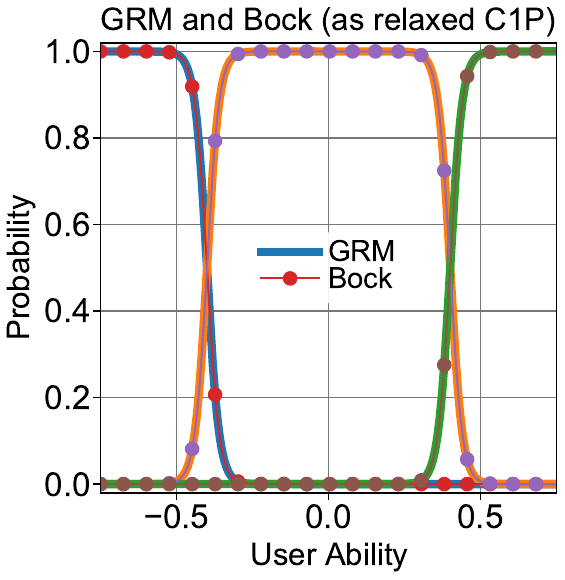}	
	\caption{}
	\label{Fig_GRM_C1P}
\end{subfigure}
\vspace{-3mm}
\caption{(a) Example illustrating that GRM, for all practical purposes 
can be seen as a special case of Bock.
Shown are GRM with $a=8$, $b_h=(-0.2, 0.2)$
and Bock with $\alpha_h=(8, 16)$, $\beta_h=(1.6,0)$.
Three curves correspond to the probability for users to choose three options respectively.
(b): Example illustrating that GRM and Bock (like all IRT models discussed in this paper) can be seen as relaxed versions of
a response matrix with the C1P property.
Shown here are GRM with $a=50$, $b_h=(-0.4, 0.4)$
and Bock with $\alpha_h=(50, 100)$, $\beta_h=(20,0)$.
Notice the similarity with \cref{fig:Fig_Intro}(b) after accounting for the linear ordering of the responses.
}
\label{Fig_GRM_vs_Bock}
\end{figure}

Notice that the logits in the cumulative response function $\P^*$ for GRM
all have the same slopes $a_i$, 
which is described as the \emph{homogeneous case} of the graded response model \cite{samejima1997graded}.
While the correspondence is not exact, GRM can be interpreted as an approximate special case of the Bock model (see \cref{Fig_GRM_vs_Bock}).

\begin{figure*}[h!]
\captionsetup[subfigure]{justification=centering}
    \centering
    \begin{subfigure}[h]{0.23\textwidth}
        \centering
        \includegraphics[width=\textwidth]{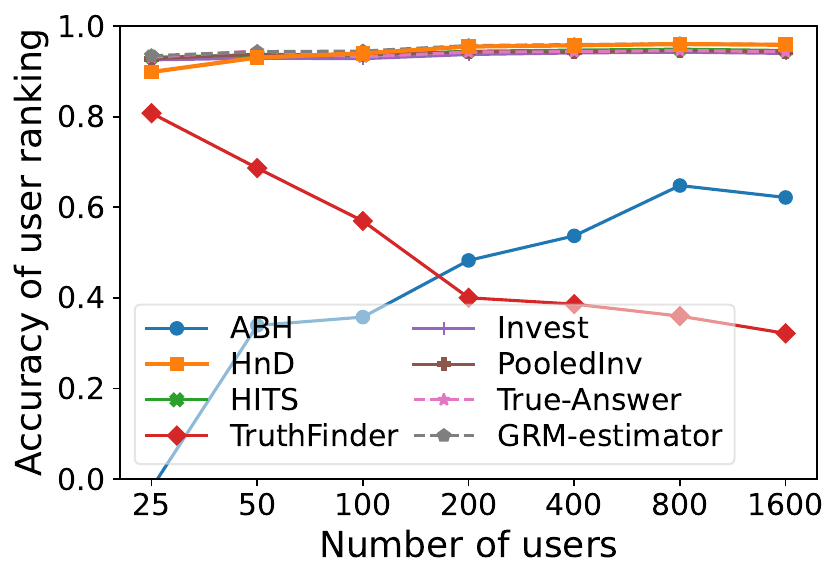}
        \caption{Varying $m$ (GRM)
		}
        \label{fig:acc_grm_u}
    \end{subfigure}
	~
    \begin{subfigure}[h]{0.23\textwidth}
        \centering
        \includegraphics[width=\textwidth]{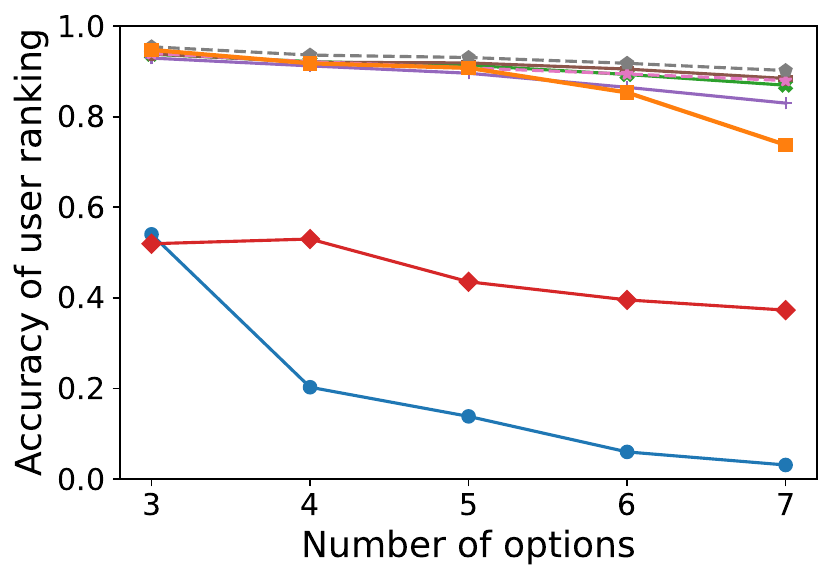}
        \caption{Varying $k$ (GRM)
		}
        \label{fig:acc_grm_o}
    \end{subfigure}
	~
    \begin{subfigure}[h]{0.23\textwidth}
		\centering
		\includegraphics[width=\textwidth]{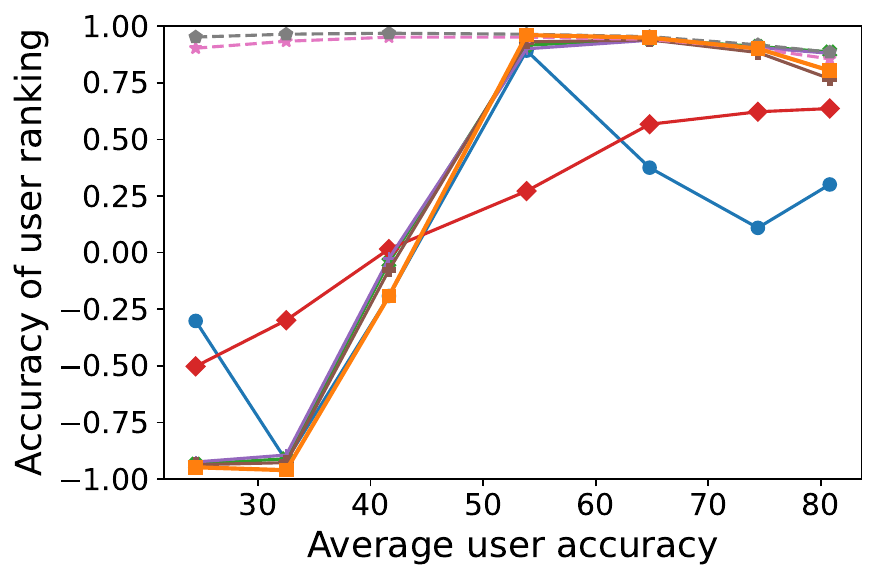}
        \caption{Varying $b_{ih}$ (GRM)}
        \label{fig:acc_grm_d}
    \end{subfigure}
    ~
    \begin{subfigure}[h]{0.23\textwidth}
		\centering
		\includegraphics[width=\textwidth]{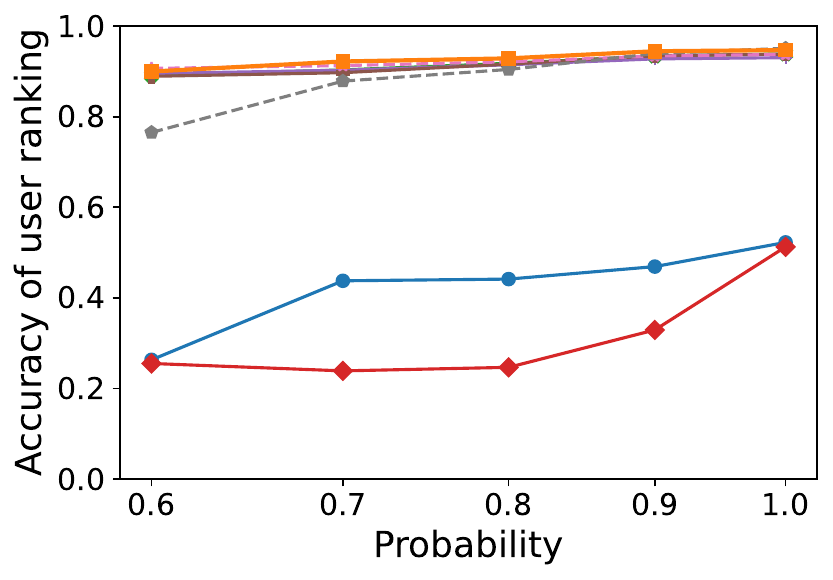}
        \caption{Varying $p$ (GRM)}
        \label{fig:acc_grm_pd}
    \end{subfigure}
    
    \begin{subfigure}[h]{0.23\textwidth}
        \centering
        \includegraphics[width=\textwidth]{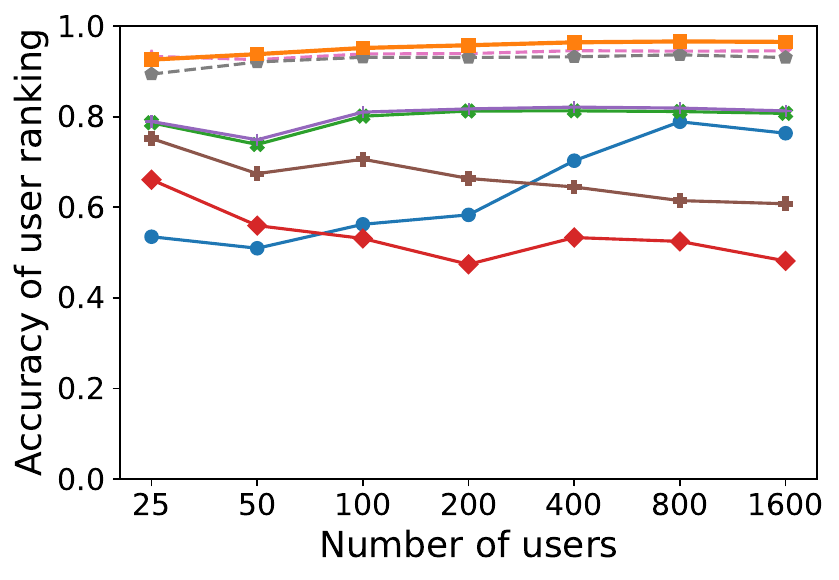}
        \caption{Varying $m$ (Bock)}
        \label{fig:acc_bock_u}
    \end{subfigure}
	~
    \begin{subfigure}[h]{0.23\textwidth}
        \centering
        \includegraphics[width=\textwidth]{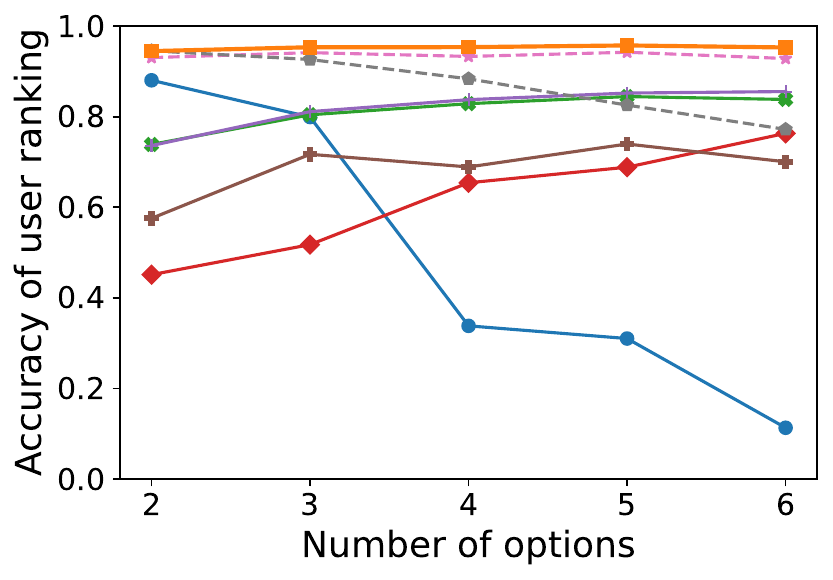}
        \caption{Varying $k$ (Bock)}
        \label{fig:acc_bock_o}
    \end{subfigure}
	~
    \begin{subfigure}[h]{0.23\textwidth}
		\centering
		\includegraphics[width=\textwidth]{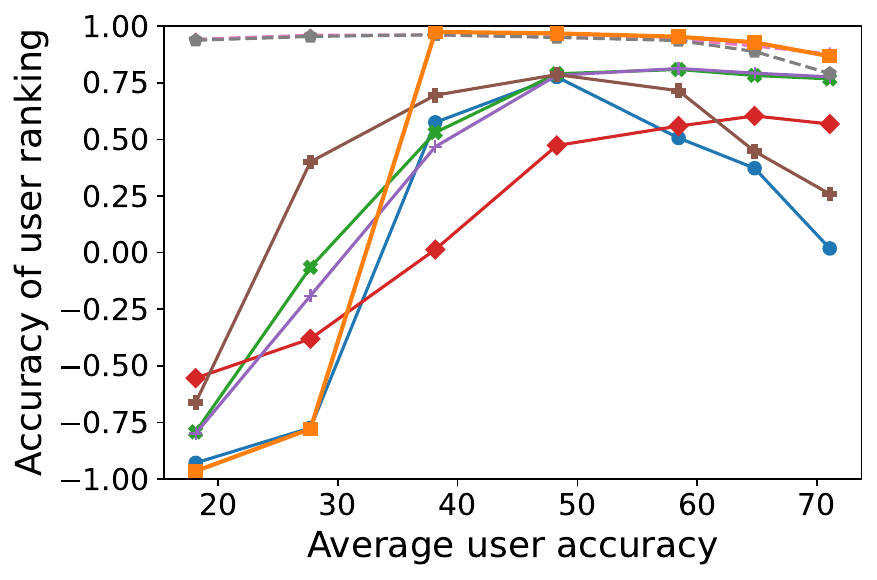}
        \caption{Varying $b_{ih}$ (Bock)}
        \label{fig:acc_bock_d}
    \end{subfigure}
    ~
    \begin{subfigure}[h]{0.23\textwidth}
        \centering
        \includegraphics[width=\textwidth]{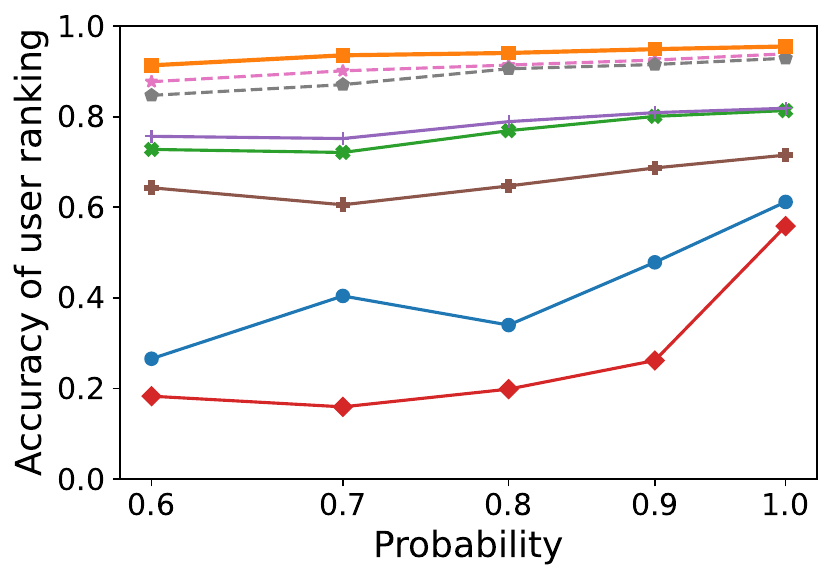}
        \caption{Varying $p$ (Bock)}
        \label{fig:acc_bock_pd}
    \end{subfigure}

    \begin{subfigure}[h]{0.23\textwidth}
        \centering
        \includegraphics[width=\textwidth]{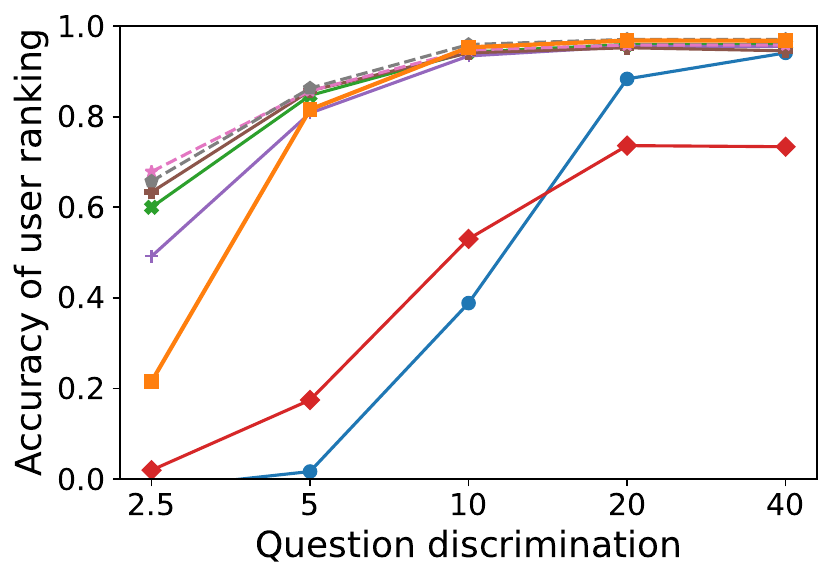}
        \caption{Varying $a_{ih}$  (GRM)}
        \label{fig:acc_grm_a}
    \end{subfigure}
	~
	\begin{subfigure}[h]{0.23\textwidth}
        \centering
        \includegraphics[width=\textwidth]{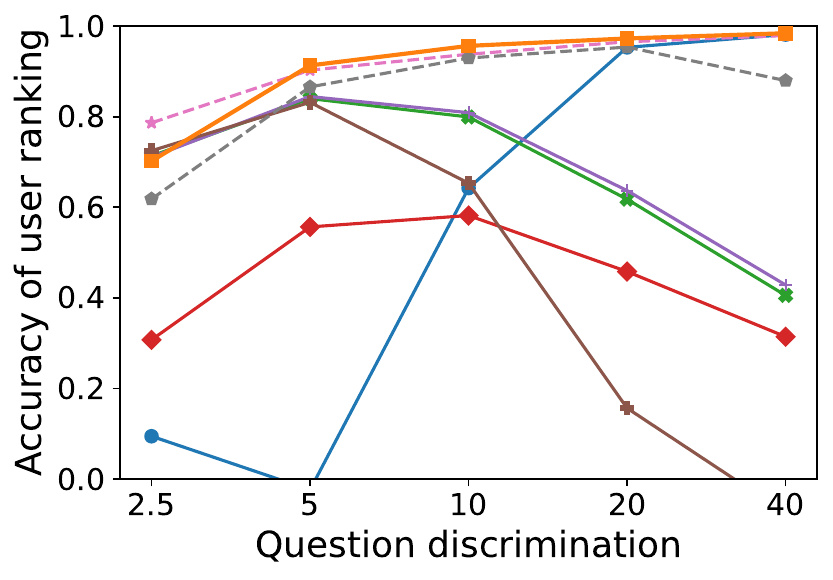}
        \caption{Varying $a_{ih}$ (Bock)}
        \label{fig:acc_bock_a}
    \end{subfigure}
	~
    \begin{subfigure}[h]{0.23\textwidth}
        \centering
        \includegraphics[width=\textwidth]{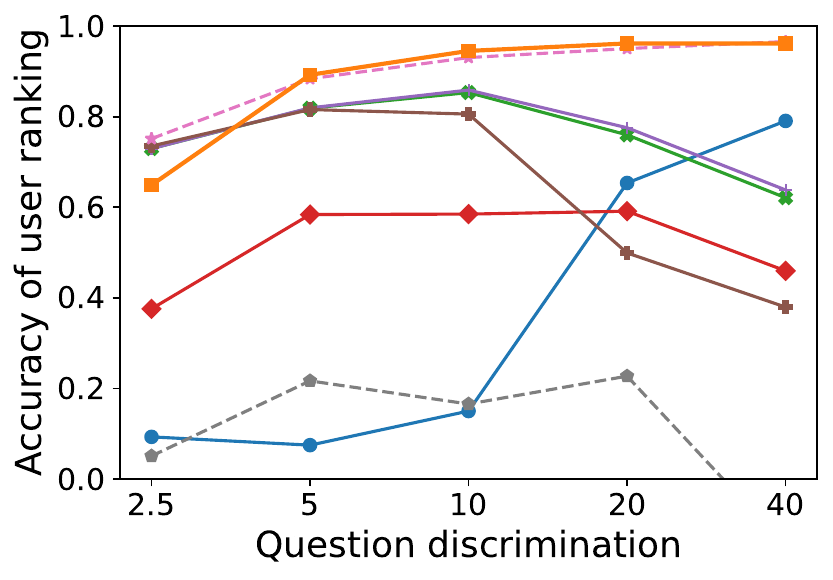}
        \caption{Varying $a_{ih}$  (Samejima)}
        \label{fig:acc_samejima_a}
    \end{subfigure}

	\caption{\Cref{sec:acc2}: Results of supplementary accuracy experiments (the legend is in the first figure).
	}
\label{fig:acc2}
\end{figure*}

\introparagraph{Samejima}
Similar to what 3PL adds to 2PL, Samejima \cite{samejima1979new} introduced a model with \emph{random guessing} based on the Bock model.
The model postulates a latent ``{don't-know}'' choice numbered zero along with the given $k$ options.
A student with low enough ability will guess randomly among the options.
Notice that this model does not exactly specialize to the 3PL model for $k=2$.
Unlike 3PL which is able to set a random guessing probability, the Samejima model fixes the probability of a student randomly picking the correct answer to be $1/k$.
The probability for a user with ability $\theta$ to pick an option $h$ for a question $i$ is given by the response function:
\begin{align*}
	\P_{ih}(\theta) = \frac{e^{\alpha_{ih}\theta+\beta_{ih}} + \frac{e^{\alpha_{i0}\theta+\beta_{i0}}}{k}}
		{\sum_{l=0}^{k}e^{\alpha_{il}\theta+\beta_{il}}}
\end{align*}
For $\beta_0 \rightarrow -\infty$, this model recovers the Bock model.

\section{Supplementary experiments}
\label{app:more_experiments}
\subsection{Accuracy on synthetic data}
\label{sec:acc2}
In addition to all experiments we show in \cref{sec:acc}, we also conduct an experiment to determine the accuracy as function of option discrimination $a_{ih}$.
Moreover, in \cref{sec:acc}, we focus on the most general polytomous IRT model, Samejima to avoid redundancy.
Here, we also present the results on the other two models GRM and Bock based on the same setup we use in \cref{sec:acc}.

\introparagraph{Supplementary accuracy experiments on Bock and GRM data (\Cref{fig:acc_grm_u,fig:acc_grm_o,fig:acc_grm_d,fig:acc_grm_pd,fig:acc_bock_u,fig:acc_bock_o,fig:acc_bock_d,fig:acc_bock_pd})}
In \cref{sec:acc}, \cref{fig:acc_grm_q,fig:acc_bock_q} show the robust performance of \HnD for datasets generated by Bock and GRM model with varying $n$ but for other experiments we focus on the most general model, Samejima.
Here we provide supplementary experimental resuts on data generated by Bock and GRM model with the same setup as what we do on data generated by the Samejima modell in \cref{sec:acc} (corresponding to \cref{fig:acc_samejima_u,fig:acc_samejima_o,fig:acc_samejima_d,fig:acc_samejima_p}).
The experimental results generally fits our observation in \cref{sec:acc}.

One major exception appears in \cref{fig:acc_grm_d,fig:acc_bock_d}.
There we see that when the average accuracy is low for Bock and GRM model, \HnD gives the approximate reverse ranking.
The reason is when the difficulty is larger than the user ability, 
good users do not return consistent correct answers but worse users all choose the worst answer for such models without random guessing. 
Thus all the methods tend to believe the majority in such arguably unrealistic setup.
Still, we see that for Samejima, \HnD performs robustly well because Samejima models random guessing
so when the difficulty is high, the bad users tend to answer random option so all methods know how to find the good users.
In all the realistic scenarios, we see \HnD outperforms other competitors.

\introparagraph{Varying question discriminations $a_{ih}$
(\Cref{fig:acc_grm_a,fig:acc_bock_a,fig:acc_samejima_a})}
Recall from \Cref{sec:irt} that larger discriminations imply that a question can more easily
``discriminate'' between able and less able users, and the correct option can more easily be picked by good users.
Here, we change the discrimination range from the default $[0,a_{\max}]$ 
with $a_{\max}=10$
to 5 different ranges, $a_{\max} \in \{2.5, 5, 10, 20, 40\}$.
We see that \HnD keeps high accuracy except when $a_{\max} = 2.5$.
When $a_{\max}$ is small, the performance between good users and bad users are close, which means the dataset is not good at evaluating users.

\subsection{Accuracy on real-world data}
Here we include the details of the real datasets used in \cref{tab:data} 
and the experimental results on each individual dataset of \cref{fig:multichoice}.\footnote{For the first two experiments, the correlation of ABH is actually negative. We use its absolute value for the simplicity of presentation.}

\begin{figure}[t]
\setlength{\tabcolsep}{3pt}	
	\centering
	\begin{tabular}{|c|c|c|c|c|}
	\hline
	Dataset & \#users & \#questions & \#options & result \\
	\hline\hline
	Chinese  & 50 & 24 & 5 & \cref{fig:multi_1} \\
	\hline
	English & 63 & 30 & 5 & \cref{fig:multi_2} \\
	\hline
	IT & 36 & 25 & 4 & \cref{fig:multi_3} \\
	\hline
	Medicine & 45 & 36 & 4 & \cref{fig:multi_4}\\
	\hline
	Pokemon & 55 & 20 & 6 & \cref{fig:multi_5}\\
	\hline
	Science & 111 & 20 & 5 & \cref{fig:multi_6}\\
	\hline
	\end{tabular}
\caption{Summary of real datasets}
\label{tab:data}
\end{figure}

\begin{figure}[t]
\captionsetup[subfigure]{justification=centering}
    \centering
    \begin{subfigure}[h]{0.23\textwidth}
        \centering
        \includegraphics[width=\textwidth]{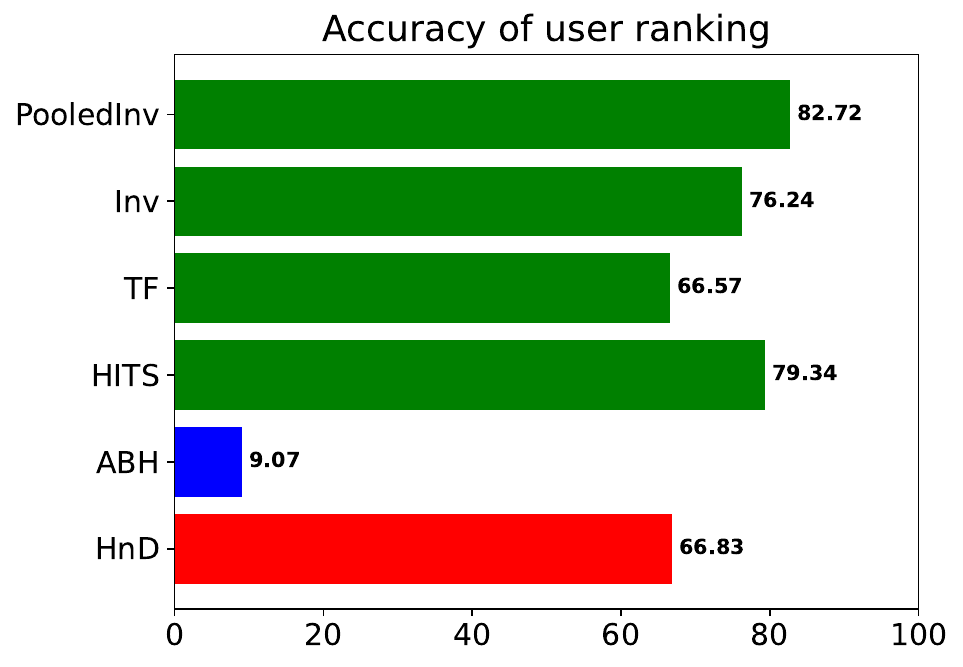}
        \caption{Chinese}
        \label{fig:multi_1}
    \end{subfigure}
	~
    \begin{subfigure}[h]{0.23\textwidth}
        \centering
        \includegraphics[width=\textwidth]{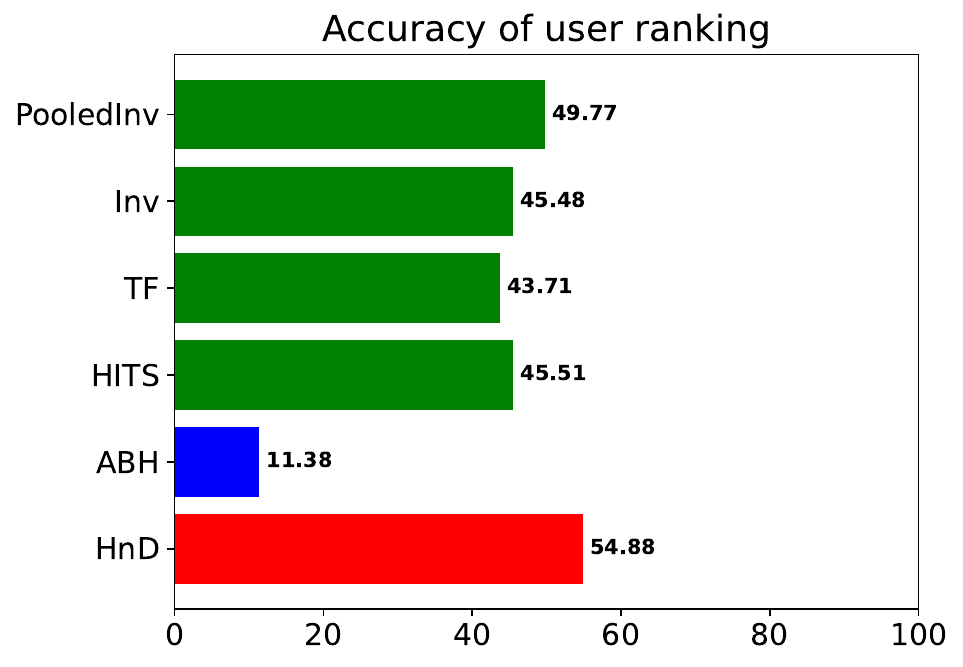}
        \caption{English}
        \label{fig:multi_2}
    \end{subfigure}
    
     \begin{subfigure}[h]{0.23\textwidth}
        \centering
        \includegraphics[width=\textwidth]{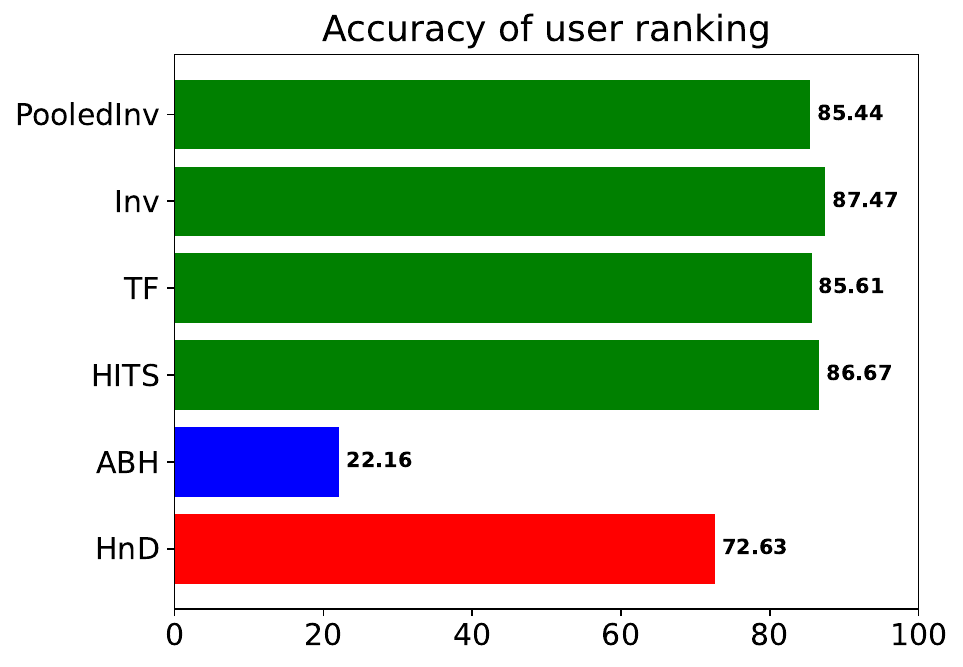}
        \caption{IT}
        \label{fig:multi_3}
    \end{subfigure}
	~
    \begin{subfigure}[h]{0.23\textwidth}
        \centering
        \includegraphics[width=\textwidth]{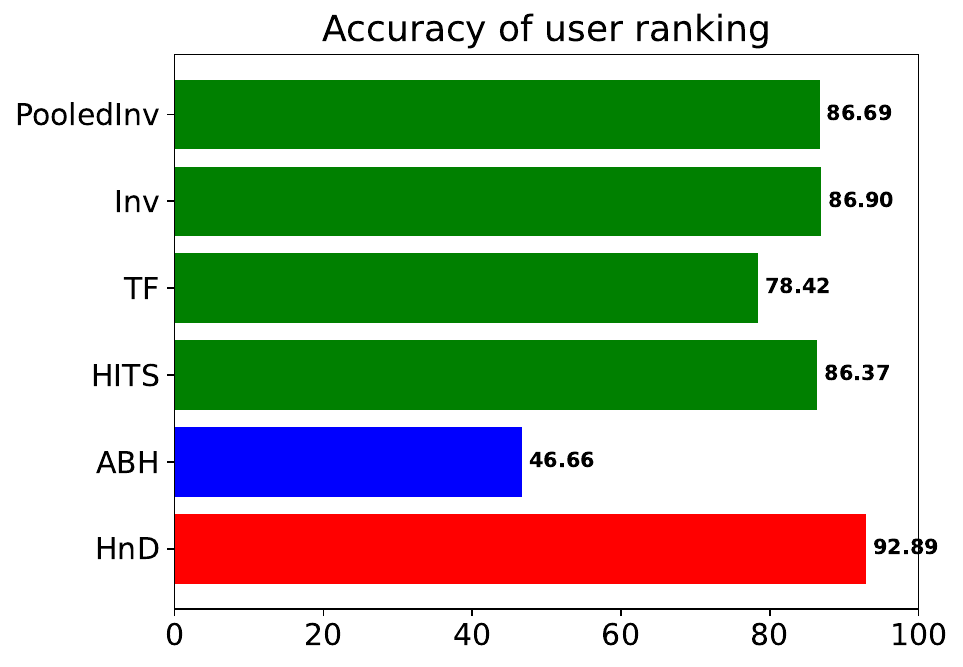}
        \caption{Medicine}
        \label{fig:multi_4}
    \end{subfigure}
    
     \begin{subfigure}[h]{0.23\textwidth}
        \centering
        \includegraphics[width=\textwidth]{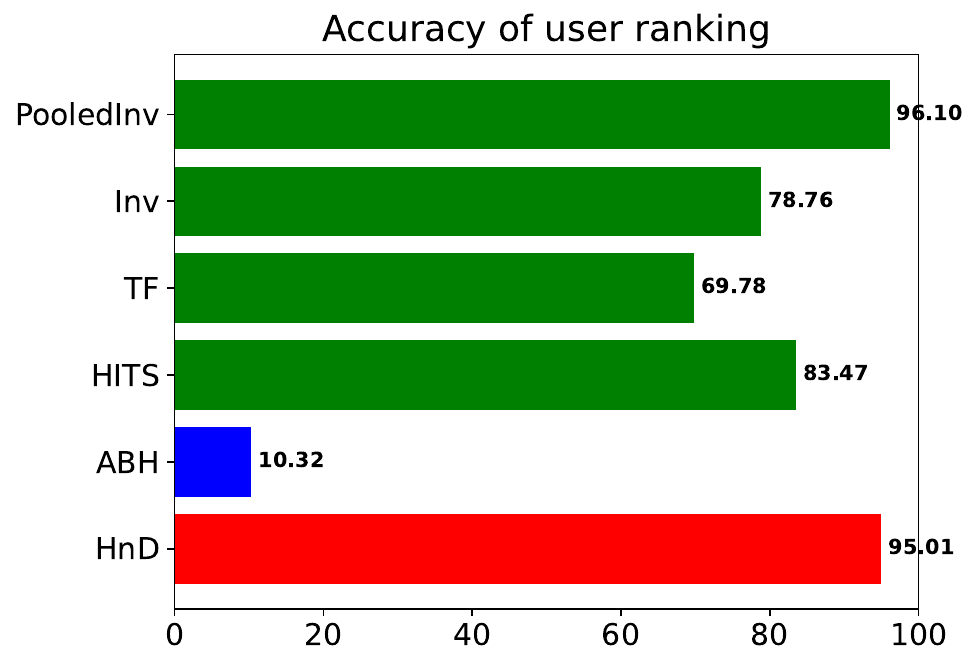}
        \caption{Pokemon}
        \label{fig:multi_5}
    \end{subfigure}
	~
    \begin{subfigure}[h]{0.23\textwidth}
        \centering
        \includegraphics[width=\textwidth]{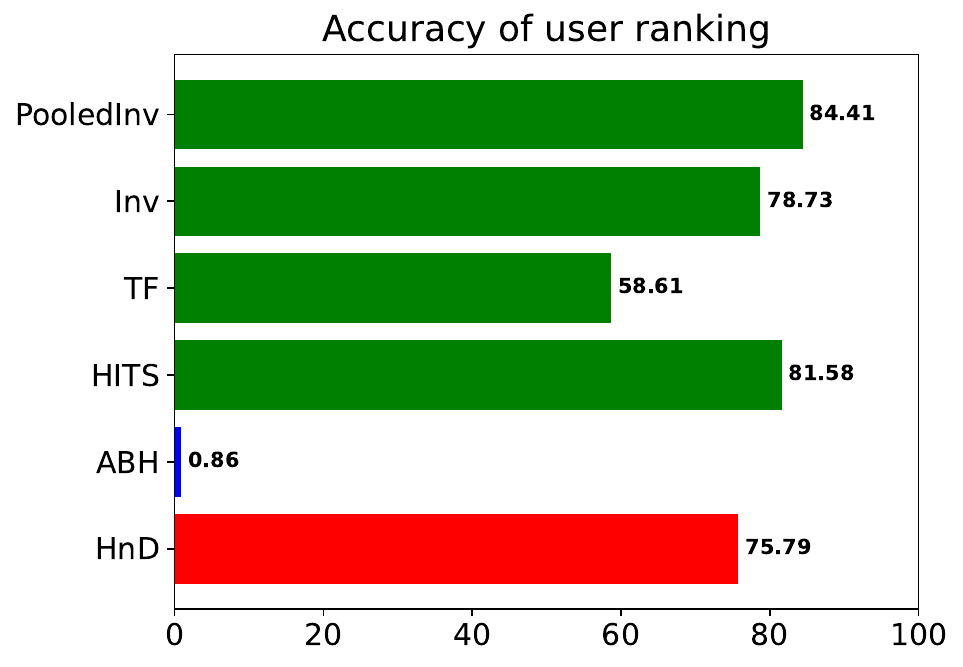}
        \caption{Science}
        \label{fig:multi_6}
    \end{subfigure}
\caption{Detailed result on each dataset in \cref{fig:multichoice}}
\label{fig:detail}
\end{figure}

We like to remind that these datasets \emph{do not come with a ground truth on user ranking} 
and we instead inferred an approximate ranking
by using the ``True answer''  baseline explained in \cref{sec:benchmark} as the reference ranking against which to compare. 
We have seen in the earlier experimental results that sometimes \HnD gives better answers than ``True answer.''
Thus the experimental results should be seen only  as a supplementary evidence.

\subsection{Accuracy on realistic simulated data}
\label{sec:realdata_appendix}
As mentioned in \cref{sec:benchmark}, we do not know \emph{any existing benchmark with a known true ranking} of users by their abilities.
In order to still verify the performance of \HnD in an as-realistic-as-possible setting, 
we simulate real-world data by generating synthetic data that follow prior published IRT parameters of real-world data ~\cite{demars} or estimated distribution \cite{VaniaHHMPPLCB20}.
In this way, we have realistic data with the ground truth.

\introparagraph{Experiments on simulated American Experience test data}
DeMars~\cite{demars} presents a detailed analysis on the American Experience test
with 40 questions and 2692 participating students. 
On page 87, the book presents estimates of parameters for a binary 3PL IRT model.
We generate 10 small datasets with 100 students and 40 questions, fitting the size of a real-world class and 10 large datasets with 2692 students and 40 questions following the original dataset size provided in \cite{demars}.
In datasets of both settings, the 40 questions follow exactly those estimated parameters from the book to simulate a realistic scenario, and the student (user) abilities follow the normal distribution $N(0, 1)$ as indicated in \cite{demars}.
We then conduct experiments on each of them and take the average accuracy.
Our results in \cref{fig:simulated} clearly show the stronger and more stable performance of \HnD compared to other competitors except for the cheating estimators.

\begin{figure}[t]
\centering
\begin{subfigure}[h]{0.23\textwidth}
	\centering
	\includegraphics[width=\textwidth]{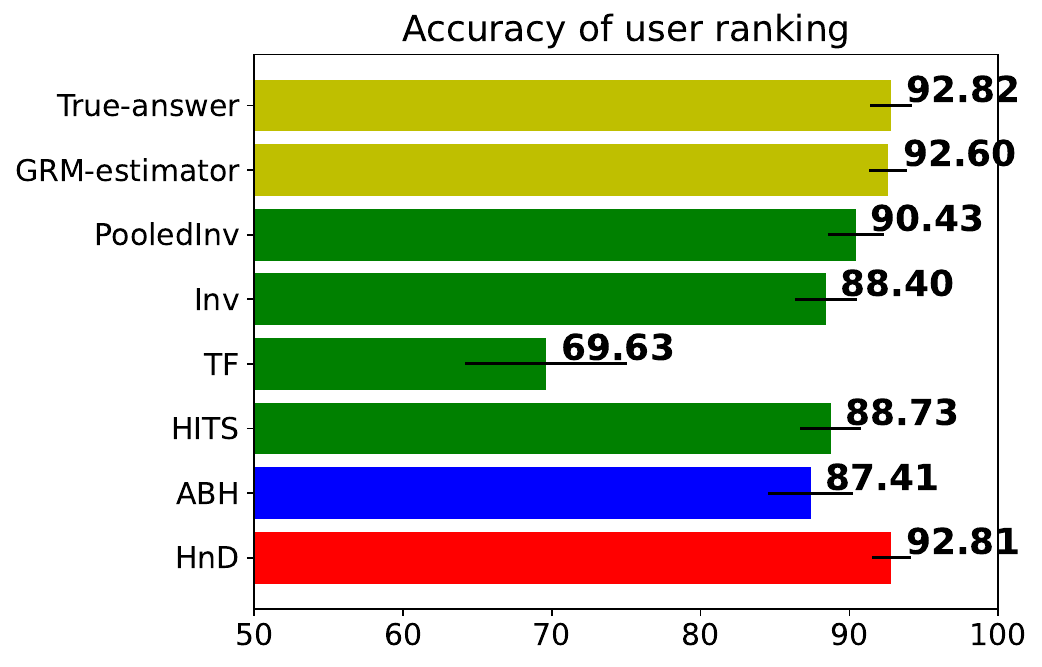}
	\caption{100 students}
	\label{fig:simulated_small}
\end{subfigure}
~
\begin{subfigure}[h]{0.23\textwidth}
	\centering
	\includegraphics[width=\textwidth]{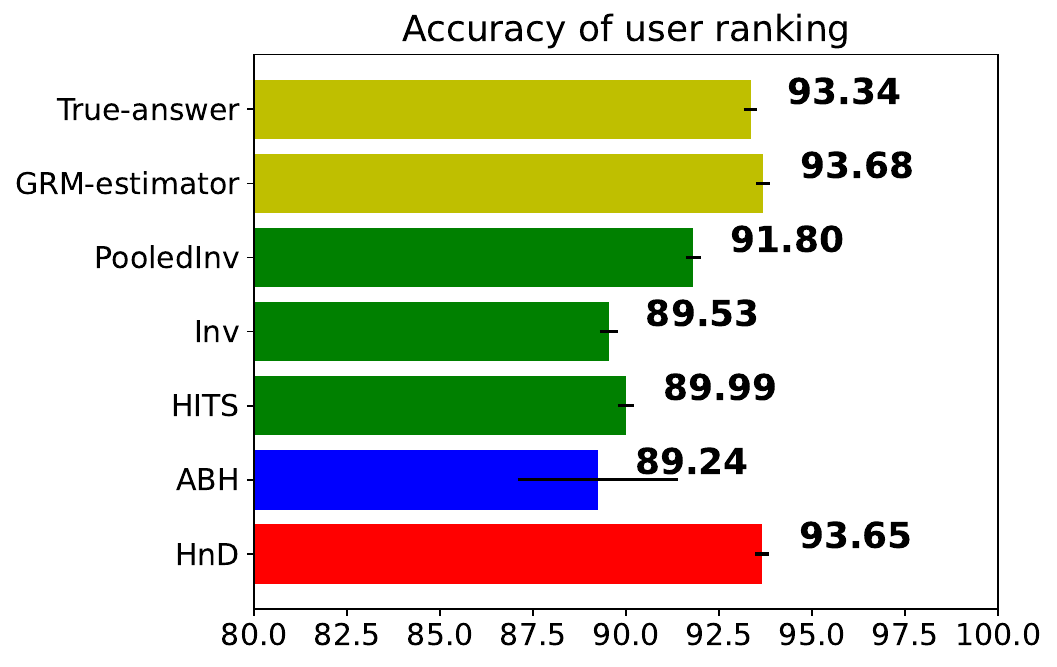}
	\caption{2692 students}
	\label{fig:simulated_large}
\end{subfigure}
\caption{\cref{sec:realdata_appendix}: Accuracy experiments on synthetically simulated data
with correlation of user rankings following parameters from real world data described in \cite{demars}.
Shown are average and standard deviation over 10 runs.
}
\label{fig:simulated}
\end{figure}

\introparagraph{Experiments on simulated half-moon data}
Besides the simulated data based on the data distribution of the American Experience test, which exactly fits our background of students and questions, we also explore the performance of \HnD on other data.
Prior work~\cite{VaniaHHMPPLCB20} provides a detailed analysis of IRT parameters over 29 datasets.
They consider different natural language understanding models as users in our scenario and use the performances of those models to estimate the binary 3PL IRT parameters
(\cref{sec:IRT_binary})
of the tasks (questions in our scenario) in the datasets.
One important take-away from the paper is that the distribution of log discrimination versus difficulty tends to have a half-moon pattern, which means most of the discriminative questions are either easy or difficult.
Following this, we develop a simulated question generator 
with the half-moon distribution of log discrimination versus difficulty as shown in \cref{fig:distribution}.
As for other parameters, we set user ability to follow the normal distribution $N(0, 1)$
and random guessing $c$ to be within $[0,0.5]$ as hinted by \cite{VaniaHHMPPLCB20}.
We generate 10 datasets with 100 users and 100 questions, 
conduct experiments on them and take the average accuracy 
as the experimental result.
\Cref{fig:simulated} presents the strong performance of \HnD, which shows that \HnD may be also applicable to other crowdsourcing tasks.

\begin{figure}[h!]
	    \captionsetup[subfigure]{justification=centering}
	    \centering
	    \begin{subfigure}[h]{0.18\textwidth}
	        \centering
	        \includegraphics[width=\textwidth]{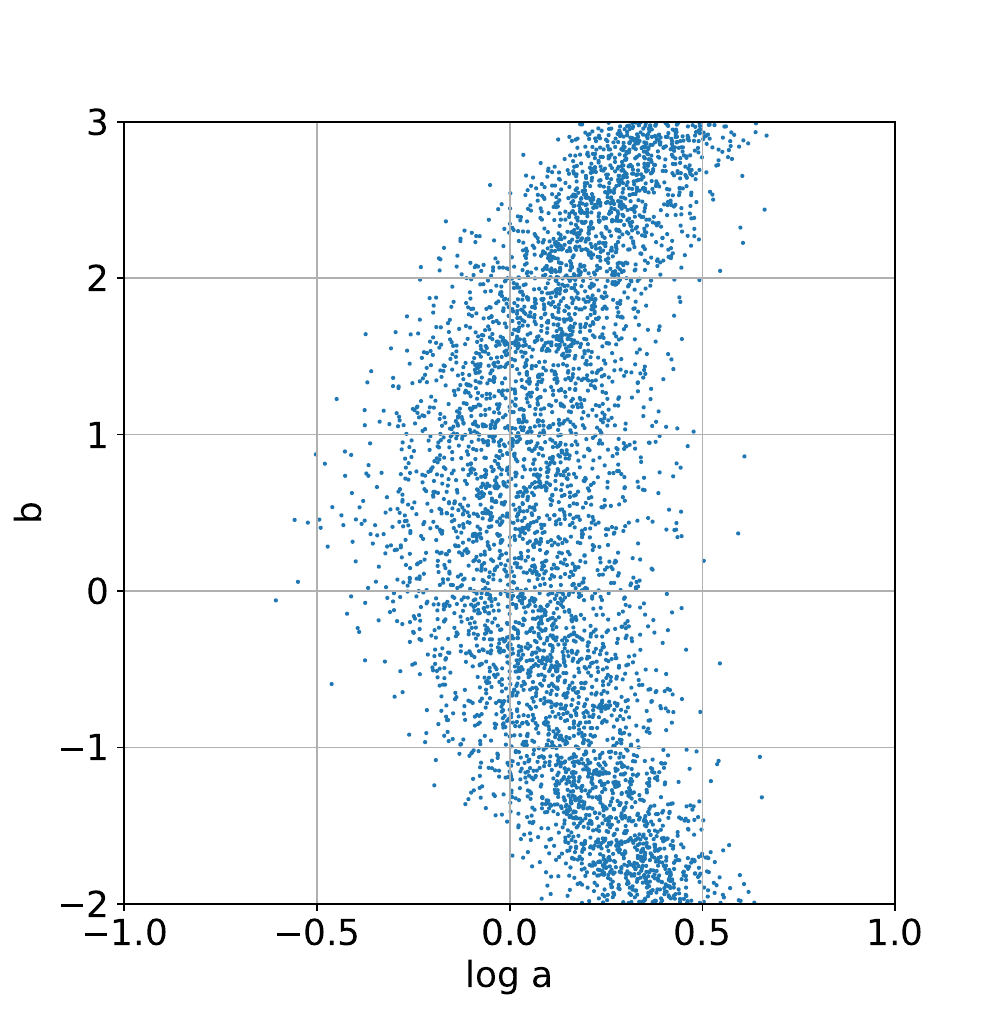}
	        \caption{Distribution of simulated data with half moon shape
			}
	        \label{fig:distribution}
	    \end{subfigure}
		~
	    \begin{subfigure}[h]{0.28\textwidth}
	        \centering
	        \includegraphics[width=\textwidth]{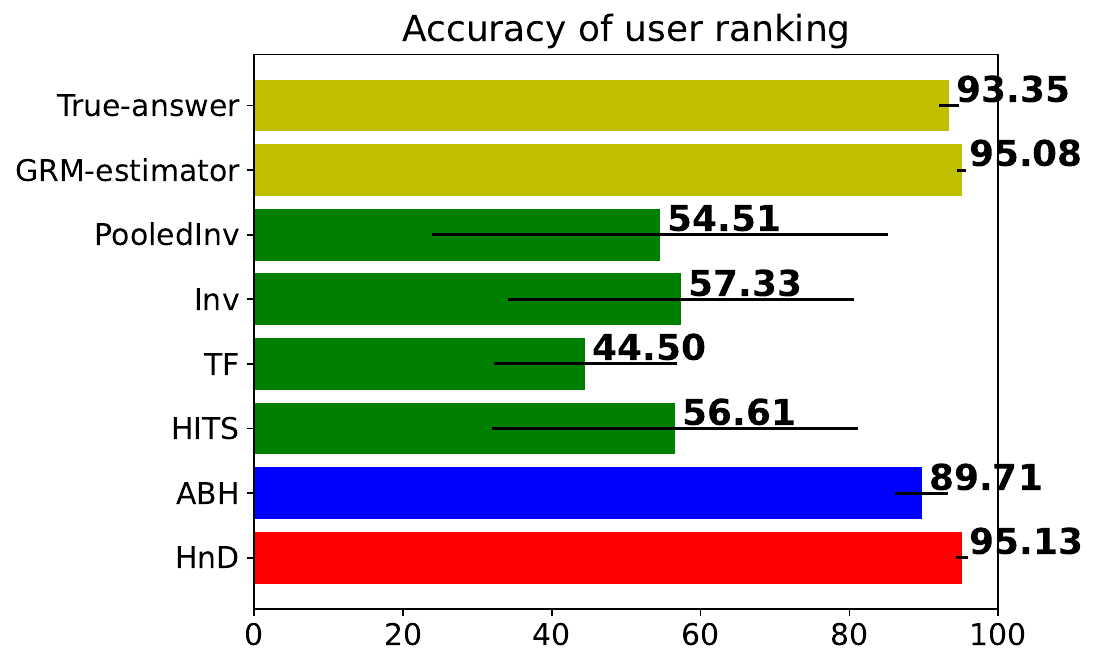}
	        \caption{Correlation of user ranking on simulated data following the half moon data distribution}
	        \label{fig:simulated}
	    \end{subfigure}
	    \caption{Experiments on the simulated half moon data}
\end{figure}

\subsection{Implementation details}
To get an approximate correspondence in discrimination between Bock and GRM models, 
when we parameterize $a_{ih}$ for Bock from $[0, x]$,
we also parameterize $a_{i}$ for GRM from 
$[0, \frac{2x}{k+1}]$.
In this way, the average $a_{ih}$ for Bock is $\frac{x}{2}$ 
and the average GRM's $a_{i}$ is $\frac{x}{k+1}$.
As discussed in \cref{sec:irt}, 
$a_{i}\!=\!\frac{x}{k+1}$ for GRM 
is approximately the same case as 
$a_{ih} \!=\! \frac{hx}{k+1}$ for Bock for each $h$.
This means $a_{i} \!=\! \frac{x}{k+1}$ for GRM 
approximates the case that the average 
$a_{ih}$ for Bock equals 
$(\sum \frac{hx}{k+1})/k \!=\! \frac{(k+1)kx}{2(k+1)}/k \!=\! \frac{x}{2}$.
In this way, we get similar average discriminations for Bock and GRM.

Another minor detail is that when we generate the C1P data, directly applying the distribution where all user abilities are distributed evenly within $[0,1]$ will result in completely symmetric response matrix and thus no method (even with symmetry breaking) could possibly determine the side with better users.
Therefore we set 10\% of the users to be within $[0,0.5]$ and  90\% users to be within $[0.5,1]$ without influencing the property that the response matrices are consistent.

\section{Detailed Analysis}
\subsection{Dawid-Skene (DS) for homogenous items}
\label{sec:DS}
Besides the Item Response Theory, another widely used model for item labeling was proposed by Dawid and Skene (DS)~\cite{dawid}.
Here we first introduce the DS model and compare it against thhe IRT models.

The DS model assumes \emph{homogenous items},
i.e.\ they are of the same type.
It is parameterized by one latent stochastic confusion matrix per user where
off-diagonal entries represent the probabilities that a user mislabels an item from one class with another while the diagonal elements correspond to probabilities of making accurate choices. 
For example, assume $m=5$ users have to choose one of $k=3$ labels dog, cat, or rabbit for each of $n=10$ images.
One entry in the confusion matrix for a user specifies the probability that that user chooses the label `dog' 
when the true label is `cat.'
The approach assumes \emph{identical parameters for each item} and thus requires $mk(k-1)$ parameters for $m$ users.
The confusion matrices and true labels are jointly estimated by maximizing the likelihood of observed labels
via the EM algorithm~\cite{dawid}.

This model is widely used and studied in crowdsourcing applications for homogenous questions
\cite{DBLP:journals/tkde/HungVTWYZ18,DBLP:journals/jmlr/ZhangCZJ16,DBLP:conf/nips/LiuPI12,DBLP:journals/jmlr/RaykarYZVFBM10},
and a recent survey~\cite{Zheng:2017:TIC:3055540.3055547} on crowdsourcing
recommends it as an implementation with little overhead.
However, DS is not suited to model a scenario where \emph{every question can be different}, such as in multiple-choice questions testing student understanding.
IRT, in contrast to DS, uses one (or more) latent parameters \emph{for each question}.

\subsection{Detailed Analysis on ABH-power}
\begin{figure}[t]
\captionsetup[subfigure]{justification=centering}
    \centering
    \begin{subfigure}[h]{0.23\textwidth}
        \centering
        \includegraphics[width=\textwidth]{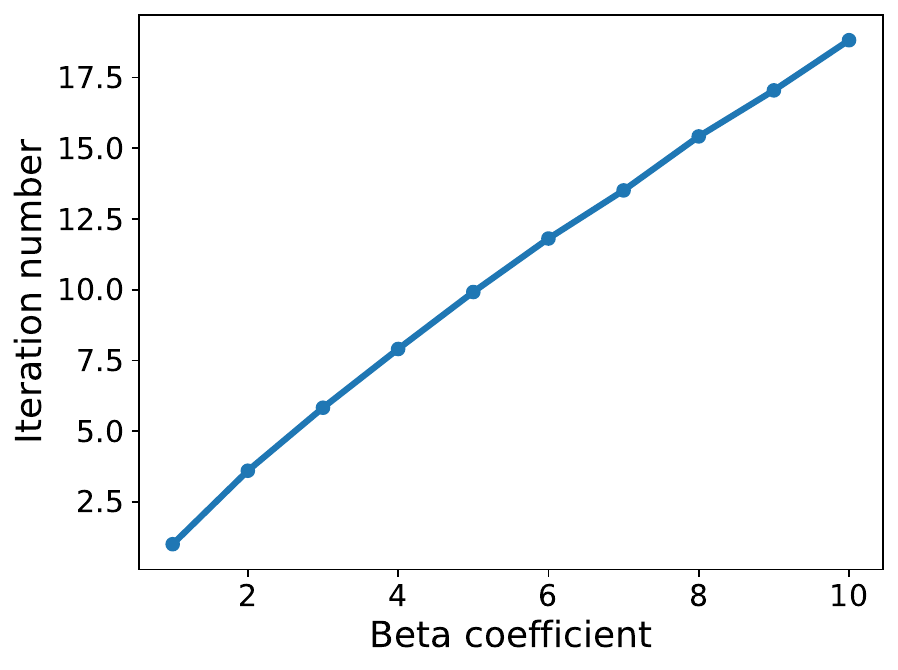}
        \caption{The influence of $\beta$. Each iteration number is divided by the smallest number.}
        \label{fig:beta}
    \end{subfigure}
	~
    \begin{subfigure}[h]{0.23\textwidth}
        \centering
        \includegraphics[width=\textwidth]{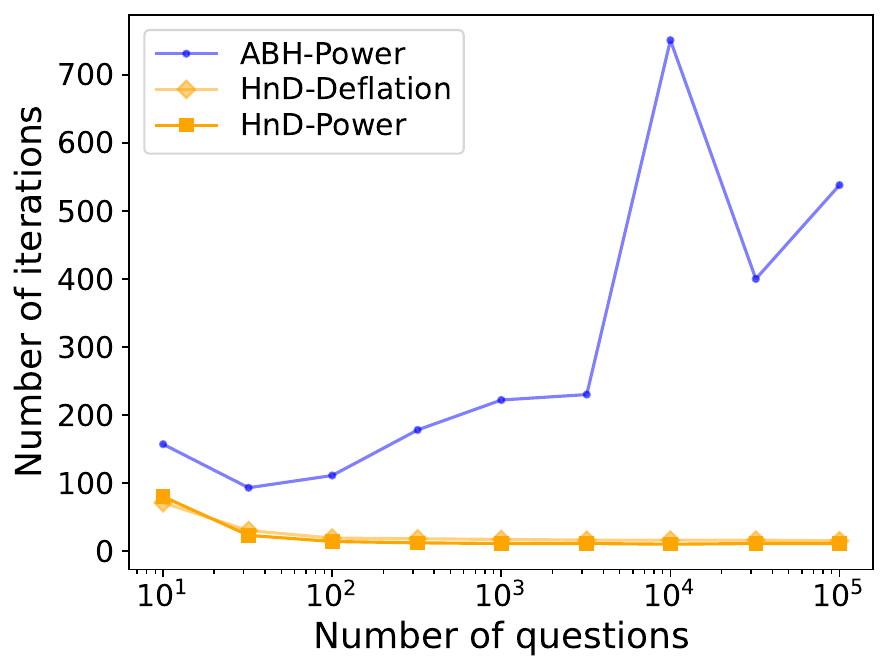}
        \caption{Iteration numbers with varying $n$}
        \label{fig:iteration_q}
    \end{subfigure}
\caption{Detailed Analysis on ABH-power}
\label{fig:detail_ABH}
\end{figure}

As discussed in \cref{sec:complexity}, it is possible theoretically to implement ABH in a similar way as \HnD whose time complexity is $\O(mnt)$.
However, instead of $\O(mnt)$, the time complexity of ABH-Power is $\O(mnt + m^2t)$, which can be larger than $\O(mnt)$ when $m >> n$.
We start from the algorithm to analyze ABH-Power in detail.

\introparagraph{The algorithm of ABH-Power}
\cref{alg:abh} presents our novel approach for implementing ABH which we call ABH-Power.
The main difference from \cref{alg:hnd} are \cref{alg:a5,alg:a6}.
Note that $\response$ is a $(m \times n)$ matrix but both $D$ and $I$ are $(m \times m)$ matrices which are much larger than $\response$ when $m$ is much larger than $n$.
In this case, ABH-Power has time complexity $\O(m^2t)$, which corresponds to the quadratic execution time in \cref{fig:sca_u}.
However, this does not explain why ABH-Power is also slower in \cref{fig:sca_q}.
Besides, we observe that when question number is large (e.g., 10000) and user number is small (e.g., 100), each iteration of ABH-Power and \HnD takes similar time.
We believe that the reason for ABH-Power being slower also in \cref{fig:sca_q} is due to the required number of iterations until convergence.
We discuss this next:

\introparagraph{The influence of $\beta$}
We discussed in \cref{sec:Hnd_vs_ABH_theory} the comparison between \HnD and ABH.
To implement ABH-power, we need to rely on the matrix $\beta \vec I_{m-1} - \vec M$, and $\beta$ needs to be no smaller than all the entries and eigenvalues of $\vec M$.
In practice, we set $\beta$ to be the largest entry in the diagonal matrix of $\response\response^\transpose$ to fulfill the requirements.
Therefore, in practice, $\beta$ is very large.
\cref{fig:beta} shows an interesting study on the choice of $\beta$, which shows that the iteration number of ABH-power is linear in terms of $\beta$, which means that the large $beta$ we use results in more iterations compared to \HnD.

\introparagraph{Number of iterations}
The above analysis shows us why ABH can be much slower than \HnD.
However, to fully understand why in both \cref{fig:sca_q}, it seems ABH-Power takes longer than linear execution time, we need more detailed evidence of the number of iterations.
\cref{fig:iteration_q} shows the numbers of iterations against varying number of users and questions respectively, corresponding to \cref{fig:sca_q}.
Each shown data point is also the median iteration number over the 5 runs.
\cref{fig:iteration_q} shows that generally as the number of questions increase, the number of iterations increase as well.
This explains why in \cref{fig:sca_q}, ABH-Power is not linear.

\begin{algorithm}[t]
\caption{ABH-Power}
\label{alg:abh}
\SetKwFor{ForAll}{forall}{do}{endfch}	
\KwIn{Response matrix $\response$, randomly initialized user scores $\userscore_0$}
\KwOut{User scores $\userscore$}
\BlankLine
$\vec{s}^{\textrm{diff}} \leftarrow \vec{s}^{\textrm{diff}}_0$ 
\hspace{15.8mm}// initialize user score differences \;
\Repeat{convergence or iteration limit}
	{
		$\vec s  \leftarrow \vec T \vec{s}^{\textrm{diff}} $	\label{alg:a3} 
		\hspace{11.1mm}// update user scores \;
		$\qweight \leftarrow \response^\transpose  \userscore$	\label{alg:a4}	
		\hspace{10.5mm}// update option weights \;
		$\userscore \leftarrow \vec D \userscore - \response  \vec \qweight$	\label{alg:a5}		
		\hspace{4.5mm}// update user scores \;
		$\vec{s}^{\textrm{diff}}  \leftarrow \beta \vec I \vec{s}^{\textrm{diff}} - \vec S \userscore$ \label{alg:a6}	
		\hspace{-0.3mm}// update user score differences \;
		Normalize $\vec{s}^{\textrm{diff}}$ to be a unit vector\;
	}
$\vec s  \leftarrow \vec T \vec{s}^{\textrm{diff}} $
\end{algorithm}

\section{Reproducibility}
In this subsection, we briefly discuss our code repository structure on github~\cite{HITSnDIFFS-code:2023}
(see \cref{fig:screenshot}).
A more detailed introduction can be found in the \emph{Readme.md} file in the repository.

\begin{figure}[h]
	\centering
        \includegraphics[width=0.2\textwidth]{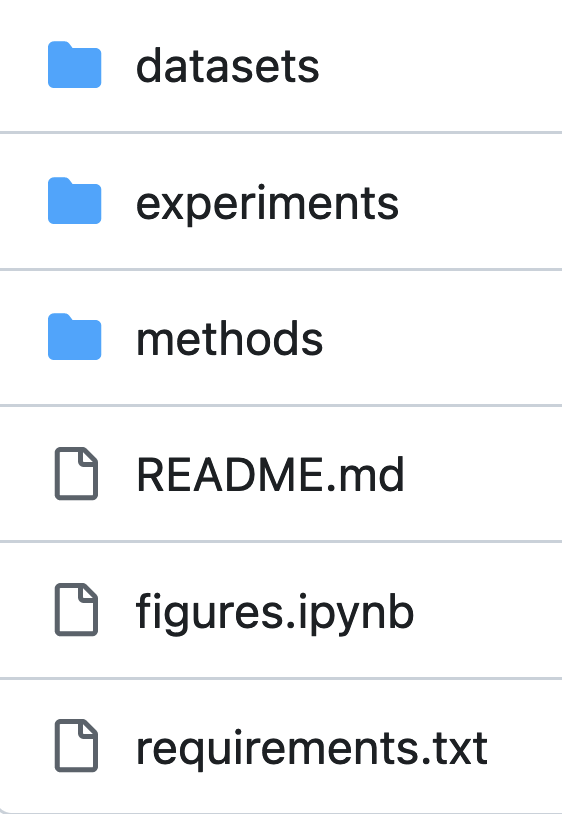}
        \caption{Repository structure}
        \label{fig:screenshot}
\end{figure}

The \emph{methods} folder contains all the methods including various implementations of HITSnDIFFs and ABH, majority vote, ``True-Answer" baseline, HITS-based approaches and GRM-estimator using the girth package.
The \emph{experiments} folder contains all the experiments reported in the paper including various accuracy and efficiency experiments on synthetic datasets, experiments on real-world datasets and experiments to compare HnD and ABH. 
The \emph{datasets} folder contains the real-world datasets we use in the experiments.
\emph{Readme.md} describes the repository structure.
With the help of \emph{requirements.txt}, all necessary python packages can be installed.
\emph{figures.ipynb} provides a notebook to reproduce every figure of experimental results in this paper.

\end{document}